\documentclass[onecolumn,draftcls]{IEEEtran}

\usepackage{cite}
\usepackage{enumitem}
\usepackage[utf8]{inputenc} % allow utf-8 input
\usepackage[T1]{fontenc}    % use 8-bit T1 fonts
\usepackage{hyperref}       % hyperlinks
\usepackage{url}            % simple URL typesetting
\usepackage{booktabs}       % professional-quality tables
\usepackage{amsfonts}       % blackboard math symbols
\usepackage{nicefrac}       % compact symbols for 1/2, etc.
\usepackage{microtype}      % microtypography

\usepackage{pgfplots}
  \pgfplotsset{compat=newest}
  \usetikzlibrary{plotmarks}
  \usetikzlibrary{arrows.meta}
  \usepgfplotslibrary{patchplots}
  \usepackage{grffile}
  \usepackage{amsmath}
  \usepackage{graphicx}
  \usepackage{float}
  \usepackage{xcolor}

\usepackage{subcaption}
\usepackage{algorithm}
\usepackage[noend]{algpseudocode}
\usepackage{mathtools}
\usepackage[utf8]{inputenc}
\usepackage[english]{babel}
\usepackage{amssymb,amsmath,amsthm}

\newtheorem{theorem}{Theorem}
\newtheorem{lemma}{Lemma}

\theoremstyle{definition}

\theoremstyle{definition}

\newcommand{\overbar}[1]{\mkern 1.25mu\overline{\mkern-1.25mu#1\mkern-0.25mu}\mkern 0.25mu}

\newcommand{\ER}{{\rm ER}}
\newcommand{\PE}{{\rm PE}}
\newcommand{\peset}{{\rm pe}_{h_0,h_1}}
\newcommand{\tmul}{t_{\rm mul}}
 
\newcommand{\Ttil}{\widetilde{T}}
\newcommand{\Btil}{\widetilde{B}}

\newcommand{\Ttilmul}{\widetilde{T}_{\rm extra}}
\newcommand{\qtil}{\widetilde{q}}

\newcommand{\qtilmul}{\widetilde{q}_{\rm extra}}

\newcommand{\Ghat}{\widehat{G}}

\newcommand{\Ehat}{\widehat{E}}

\newcommand{\kbar}{\overbar{k}}

\newcommand{\dmax}{d_{\mathrm{max}}}

\newcommand{\Bernoulli}{\mathrm{Bernoulli}}
\newcommand{\Binomial}{\mathrm{Binomial}}

\newcommand{\pe}{P_{\mathrm{e}}}

\newcommand{\Etilde}{\widetilde{E}}

\newcommand{\Yv}{\mathbf{Y}}

\newcommand{\Ac}{\mathcal{A}}
\newcommand{\Bc}{\mathcal{B}}

\newcommand{\Gc}{\mathcal{G}}

\newcommand{\Tc}{\mathcal{T}}

\newcommand{\EE}{\mathbb{E}}
\newcommand{\PP}{\mathbb{P}}

\newcommand{\Lc}{\mathcal{L}}

\title{Learning Erd\H{o}s-R\'enyi Random Graphs \\ via Edge Detecting Queries}

\author{Zihan Li, Matthias Fresacher, and Jonathan Scarlett}

\begin{document}

\maketitle

\begin{abstract}
    In this paper, we consider the problem of learning an unknown graph via queries on groups of nodes, with the result indicating whether or not at least one edge is present among those nodes.  While learning arbitrary graphs with $n$ nodes and $k$ edges is known to be hard in the sense of requiring $\Omega( \min\{ k^2 \log n, n^2\})$ tests (even when a small probability of error is allowed), we show that learning an Erd\H{o}s-R\'enyi random graph with an average of $\kbar$ edges is much easier; namely, one can attain asymptotically vanishing error probability with only $O(\kbar \log n)$ tests.  We establish such bounds for a variety of algorithms inspired by the group testing problem, with explicit constant factors indicating a near-optimal number of tests, and in some cases asymptotic optimality including constant factors.  In addition, we present an alternative design that permits a near-optimal sublinear decoding time of $O(\kbar \log^2 \kbar + \kbar \log n)$.
\end{abstract}

\long\def\symbolfootnote[#1]#2{\begingroup\def\thefootnote{\fnsymbol{footnote}}\footnote[#1]{#2}\endgroup}

\symbolfootnote[0]{ 
    Z.~Li is with the Department of Mathematics, National University of Singapore (e-mail: lizihan@u.nus.edu).
    
    M.~Fresacher is with the Faculty of Engineering, Computer \& Mathematical Sciences, University of Adelaide (e-mail: matthias.fresacher@adelaide.edu.au).  This work was done during his time at the National University of Singapore.
    
    J.~Scarlett is with the Department of Computer Science \& Department of Mathematics, National University of Singapore (e-mail: scarlett@comp.nus.edu.sg).
    
    This work was supported by an NUS Early Career Research Award.}

\section{Introduction}

Graphs are a ubiquitous tool in modern statistics and machine learning for depicting interactions, relations, and physical connections in networks, such as social networks, biological networks, sensor networks, and so on.  Often, the graph is not known {\em a priori}, and must be learned via queries to the network.  In this paper, we consider the problem of graph learning via {\em edge detecting queries}, where each query contains a subset of the nodes, and the binary outcome indicates whether or not there is at least one edge among these nodes. See Section \ref{sec:related} for previous work on this problem.

An application of this problem highlighted in previous works such as \cite{Bou05} is that of learning which chemicals react with each other, using tests that are able to detect whether any reaction occurs.  Another potential application is learning connectivity in large wireless networks: Each node has a unique identifier, and in response to a query, a node sends feedback to a central unit if the query includes both itself and at least one of its neighbors.  Then, to attain the query outcome, the central unit only has to detect whether {\em any} feedback signal was received.

We consider the fundamental question of how many queries are needed to learn the graph.  Under {\em adaptive testing} (i.e., tests can be designed based on previous outcomes), this question is well-understood \cite{Joh02}, as outlined below.  However, an impossibility result of \cite{Aba18} indicates that considerably more {\em non-adaptive} tests are needed in the worst-case sense for the class of graphs with a bounded number of edges.  We show that this picture is much more positive in the {\em average-case sense} by studying the average performance with respect to Erd\H{o}s-R\'enyi graphs \cite{Bol01}.  In addition, to demonstrate that these findings are not overly reliant on the specific random graph model,  we also present similar findings assuming only bounds on the number of edges and the maximum degree (see Appendix \ref{sec:gen_d_k}).

\subsection{Related Work}  \label{sec:related}

The problem considered in this paper can be viewed as a constrained group testing problem \cite[Sec.~5.8]{Ald19}.  We highlight the most relevant group testing works throughout the paper, and here simply refer the reader to \cite{Du93} for a survey of the zero-error setting, and to \cite{Ald19} for a survey of the small-error setting (i.e., the algorithm is allowed a small probability of failure).  These settings are fundamentally different, since the number of tests in the small error setting is $O(K \log N)$ (for $K$ defectives among $N$ items), while the  zero-error criterion requires $\Omega(\min\{ N, K^2\})$ tests.

Early works on graph learning via edge detecting queries considered identifying a single edge \cite{Aig86,Aig88} and then several edges \cite{Joh02} in a slightly more general scenario where the ``defective graph'' $G$ is known to be a sub-graph of a larger graph $H$.  Several works considered specific graph classes such as matchings, stars, and cliques \cite{Gre98,Alo04,Alo05}.  We particularly highlight the work of Johann \cite{Joh02}, who gave an adaptive procedure requiring $k \log_2 \frac{|\Etilde|}{k} + O(k)$ tests, where $\Etilde$ is the set of edges in the larger graph $H$; this bound is optimal up to the $O(k)$ remainder term.  More recently, extensions to hypergraphs have also been considered \cite{Ang06,Ang08,Dya16,Aba18a}.

While the adaptive setting is well-understood, the non-adaptive setting \cite{Kam18,Aba18} and adaptive settings with limited stages \cite{Bsh15,Hwa06,Aba18} are more challenging.  We refer the reader to \cite{Aba18} for a recent survey of what is known, with a notable distinction between Monte Carlo and Las Vegas style algorithms.  We highlight that in stark contrast with the standard group testing problem, the number of {\em non-adaptive} tests required to identify arbitrary graphs with $k$ edges and $n$ nodes is at least $\Omega(\min\{k^2 \log n,n^2\})$, {\em even under the small-error criterion}.\footnote{Note that the number of items $N$ in the standard group testing corresponds to ${n \choose 2} = \Theta(n^2)$ in the graph learning problem with $n$ nodes, since {\em pairs of nodes} (i.e., potential edges) play the role of items.  See Appendix \ref{sec:differences} for a brief description of the group testing problem.}

\subsection{Contributions}

The $\Omega(\min\{k^2 \log n,n^2\})$ hardness result given in \cite{Aba18} holds with respect to worst-case graphs containing $k$ edges, which raises the question of whether some notion of {\em average-case} or further restricted graph classes can overcome this inherent difficulty.  We focus primarily on the average case with respect to the ubiquitous Erd\H{o}s-R\'enyi random graph model\footnote{More precisely, we consider the variant introduced by Gilbert \cite{Gil59}.} 
and the small-error criterion, showing that indeed the number of tests required reduces to $O(\kbar \log n)$ for graphs with an average of $\kbar$ edges, and providing fairly tight explicit constant factors.  In Appendix \ref{sec:gen_d_k}, we describe how to attain similar results for general graphs with at most $k$ edges and maximum degree $d = o(\sqrt{k})$, albeit with slightly worse constant factors.

In more detail, we show the following for Erd\H{o}s-Renyi random graphs:
\begin{itemize}[leftmargin=5ex,itemsep=0ex,topsep=0.25ex]
    \item We provide a simple algorithm-independent lower bound based on counting the number of graphs within a high-probability set;
    \item We extend the COMP, DD, and SSS decoding algorithms \cite{Cha11,Ald14a} from standard group testing to the graph learning problem, and provide upper and lower bounds on their asymptotic performance under a natural random test design.
    \item We propose a sublinear-time decoding algorithm (and its associated test design) based on the GROTESQUE algorithm \cite{Cai13}, and show that it succeeds with high probability with $O(\kbar \log^2 \kbar + \kbar \log n)$ decoding time, thus nearly matching an $\Omega(\kbar \log \frac{n^2}{\kbar})$ lower bound.
\end{itemize}

Briefly, the above-mentioned decoding algorithms are described as follows: COMP ({\em cf.}, Section \ref{sec:comp}) assumes all pairs are edges unless their nodes are both in some negative test, DD ({\em cf.}, Section \ref{sec:dd}) uses the COMP solution to identify ``possible edges'' and then declares a pair to be an edge only if it is the unique possible edge among the nodes in some test, and SSS ({\em cf.}, Section \ref{sec:sss}) solves an integer program to find the sparsest graph consistent with the test outcomes.

While the group testing algorithms themselves extend easily to our setting, their theoretical analyses require significant additional effort (see Appendix \ref{sec:differences} for further discussion).  For instance, for group testing, the analysis is symmetric with respect to any defective set of size $k$, whereas for graph learning, different graphs with a fixed number of edges can behave very differently, and even seemingly simple tasks (e.g.,  determining the probability of a positive test) become challenging.

\begin{figure}
% \begin{figure}
     \begin{centering}
        %\includegraphics[width=0.55\columnwidth]{figs/EQ_Rates.pdf}
        %\resizebox{0.55\columnwidth}{!}{%
		\input{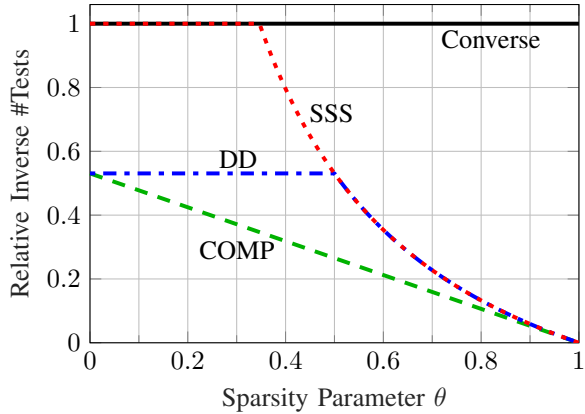}
	%}
        
    %     \par
     \end{centering}
    
    \caption{Asymptotic values of $\frac{\kbar \log_2(1/q)}{\text{\#Tests}}$ for recovering Erd\H{o}s-R\'enyi random graphs with edge probability $q = \Theta(n^{2(\theta-1)})$, and average number of edges $\kbar = q{n \choose 2}$.  The ``COMP'' and ``DD'' curves are achievability bounds, whereas the ``Converse'' and ``SSS'' curves are converse bounds for arbitrary test designs and i.i.d.~random test designs, respectively). \label{fig:rates}}
% \end{figure}
\end{figure}

With the exception of the sublinear-time decoding part, our results are summarized in Figure \ref{fig:rates}, where we plot the asymptotic ratio between the information-theoretic bound and the number of tests for each algorithm, for sparsity levels $\theta \in (0,1)$ such that $\kbar = \Theta(n^{2\theta})$.  We observe the following:
\begin{itemize}[leftmargin=5ex,itemsep=0ex,topsep=0.25ex]
    \item For $\theta > \frac{1}{2}$, the DD upper bound and SSS lower bound match under i.i.d.~random testing.  As we explain in Section \ref{sec:sss}, SSS is the optimal algorithm, so if it fails then so does any algorithm.  Hence, DD is asymptotically optimal (including constants) under i.i.d.~random testing for $\theta > \frac{1}{2}$.
    \item For $\theta \le \frac{1}{2}$, DD succeeds with fewer than twice as many tests as the optimal information-theoretic threshold; the latter is a converse bound applying to {\em any} test design (not only i.i.d.~random testing).
\end{itemize}

While analogous results have been established for standard group testing \cite{Ald14a}, we again highlight that the analysis comes with several non-trivial challenges, particularly when it comes to DD and SSS.  See Appendix \ref{sec:differences} for an outline of the main differences.

\section{Setup} \label{sec:setup}

We seek to learn an unknown undirected graph $G = (V,E)$ with $n$ nodes, i.e., the vertex set is $V = \{1,\dotsc,n\}$, and the edge set $E$ contains up to ${n \choose 2}$ pairs of nodes.  We adopt a random graph model in which each edge appears in the graph independently with probability $q$ (i.e., the Erd\H{o}s-R\'enyi graph $\ER(n,q)$).  After the graph $G$ is randomly drawn, it is fixed throughout the entire testing process (described below).

We test the nodes in groups; the output of each test takes the form
\begin{equation}
Y = \bigvee_{(i,j) \in E} \big\{ X_{i} \cap X_{j} \big\}, \label{eq:gt_noiseless_model}
\end{equation}
where the binary-valued test vector $X =(X_1,\dotsc,X_n)$ indicates which nodes are included in the test.  That is, the resulting output $Y = 1$ if and only if at least one edge exists in the sub-graph of $G$ induced by the nodes included in the test; we henceforth use the terminology that such an edge is {\em covered}. We refer to tests with $Y=1$ as positive, and tests with $Y=0$ as negative.  A total of $t$ tests are performed according to the test vectors $X^{(1)},\dotsc,X^{(t)}$ to produce the outcomes $Y^{(1)},\dotsc,Y^{(t)}$.  We focus on {\em non-adaptive} tests, where $X^{(1)},\dotsc,X^{(t)}$ must be selected prior to observing any outcomes.

Given the tests and their outcomes, a \emph{decoder} forms an estimate $\Ghat$ of the graph $G$, or equivalently, an estimate $\Ehat$ of the edge set $E$.  One wishes to design a sequence of tests $X^{(1)},\dotsc,X^{(t)}$, with $t$ ideally as small as possible, such that the decoder recovers $G$ with probability arbitrarily close to one.  The error probability is given by 
\begin{equation}
    \pe := \PP[\Ghat \ne G], \label{eq:pe}
\end{equation}
and is taken over the randomness of the graph $G$, as well as the tests $X^{(1)},\dotsc,X^{(t)}$ (if randomized).  We only consider deterministic decoding algorithms (without loss of optimality), and all of our results are asymptotic in the limit as $n \to \infty$ (with $q$ varying as a function of $n$).

\subsection{Sparsity Level}

We focus our attention on sparse graphs, i.e., $q = o(1)$ as $n \to \infty$.\footnote{In fact, the arguments given in \cite{Ald18} can be applied to the present setting to show that one essentially cannot improve on individual testing of edges when $q = \Theta(1)$.}  More specifically, we consider the sublinear scaling regime $q = \Theta(n^{-2(1-\theta)})$ for some $\theta \in (0,1)$, meaning that the average number of edges $\kbar = {n \choose 2}q$ behaves as $\Theta( n^{2\theta} )$.  By the assumption $\theta \in (0,1)$, we also have
\begin{equation}
    n^{-(2-\eta)} \ll q \ll n^{-\eta}, \quad n^{\eta} \ll k \ll n^{2-\eta} \label{eq:q_bounds}
\end{equation}
for sufficiently small (but constant) $\eta > 0$ and sufficiently large $n$.  Here and subsequently, we write $f(n) \ll g(n)$ as a shorthand for $f(n) = o(g(n))$.

\subsection{Bernoulli Random Testing}

For the most part, we will focus on the case that the tests are designed randomly: Each node is independently placed in each test with a given probability $p$.  We refer to this as {\em i.i.d.~Bernoulli testing}, or simply {\em Bernoulli testing} for short.  Analogous designs are known to lead to most of the best-known performance bounds in the group testing literature \cite{Ald14a,Sca15b}, with the exception of some slight improvements shown recently via more structured random designs \cite{Joh16,Coj19}.

We parametrize $p$ as $p = \sqrt{\frac{2\nu}{qn^2}}$ for some constant $\nu > 0$, as this scaling regime turns out to be optimal in all cases (with the choice $\nu = 1$ further being optimal for the algorithms we consider).  Note that this choice of $p$ gives $p^2 = \frac{\nu}{\kbar} (1+o(1))$, since $\kbar = \frac{1}{2}qn^2 (1+o(1))$.

When studying probabilities associated with a single random Bernoulli test, we will denote the test outcome by $Y$, and the (random) indices of nodes included in the test by $\Lc \subseteq \{1,\dotsc,n\}$.  In addition, $\PP_G[\cdot]$ denotes probability (with respect to the random testing alone) when the underlying graph is $G$.

\subsection{Typical Graphs} \label{sec:typical}

Throughout the paper, we frequently make use of the following {\em typical set} of graphs:
\begin{align}
    \Tc(\epsilon_n) = \Big\{ G \,:\, &(1-\epsilon_n)\kbar \le k \le (1+\epsilon_n)\kbar, d \le \dmax, \nonumber \\
        & (1-\epsilon_n)(1-e^{-\nu}) \le \PP_G[Y=1] \le (1+\epsilon_n)(1-e^{-\nu})  \Big\}, \label{eq:typical}
\end{align}
where $k = |E|$ is the number of edges, $d$ is the maximum degree of $G$, and
\begin{equation}
    \dmax = 
    \begin{cases}
    2nq & \theta > \frac{1}{2} \\
    \log n & \theta \le \frac{1}{2}.
    \end{cases} \label{eq:d_bound}
\end{equation}
The following lemma justifies the terminology {\em typical set} by showing that the random graph lies in this set with probability approaching one.

\begin{lemma} \label{lem:typical}
    Fix $\theta \in (0,1)$, and let $G \sim \ER(n,q)$ for some $q = \Theta(n^{-2(1-\theta)})$.  In addition, suppose that $\PP_G[Y=1]$ in \eqref{eq:typical} is defined with respect to $\Bernoulli(p)$ testing with $p = \sqrt{\frac{2\nu}{qn^2}}$ for fixed $\nu > 0$.  Then, there exists a sequence $\epsilon_n \to 0$ such that $\PP[G \in \Tc(\epsilon_n)] \to 1$ as $n \to \infty$.
\end{lemma}

The condition on $k$ in the typical set simply states that the number of edges is close to its mean, which follows by standard concentration bounds.  The bound on the maximum degree is similarly standard and straightforward to establish.  By far the most challenging part is bounding $\PP_G[Y=1]$ with high probability; this is done using the inclusion-exclusion principle (i.e., Bonferroni's inequalities) and carefully bounding the probability of a random test containing one edge, two edges, three edges, and so on.  The details are given in Appendix \ref{sec:pf_typical}.

Using the bounds on $k$ and $d$ in \eqref{eq:typical}, along with the fact that $q$ satisfies \eqref{eq:q_bounds}, we readily observe that
\begin{equation}
    d^2 \ll k, \quad dp \ll 1 \label{eq:d_vs_k}
\end{equation}
in both cases of \eqref{eq:d_bound}.  Note that the second of these statements follows immediately from the first since we focus on the regime $p = \Theta\big( \frac{1}{\sqrt k}\big)$.

\section{Algorithm-Independent Converse Bound} \label{sec:conv}

To provide a benchmark for our upper bounds, we first provide a simple algorithm-independent lower bound on the number of tests for attaining asymptotically vanishing error probability, which is based on fairly standard counting arguments and Fano's inequality \cite[Sec.~2.10]{Cov01}.

\begin{theorem} \label{thm:conv}
    Under the setup of Section \ref{sec:setup} with $q = o(1)$ and an arbitrary non-adaptive test design, in order to achieve $\pe \to 0$ as $t \to \infty$, it is necessary that
    \begin{equation}
        t \ge \bigg( \kbar \log_2 \frac{1}{q} \bigg) (1-\eta) \label{eq:t_conv} 
    \end{equation}
    for arbitrarily small $\eta > 0$.
\end{theorem}
\begin{proof}
    The proof is based on the fact that the prior uncertainty (entropy) is roughly ${n \choose 2}q\log\frac{1}{q} = \kbar \log_2 \frac{1}{q}$ bits, whereas each test only reveals one bit of information.
    See Appendix \ref{app:pf_converse} for details.
\end{proof}

% \begin{remark}
    Using a similar analysis to \cite{Bal13}, the preceding result can easily be strengthened to the {\em strong converse}, stating that $\pe$ is not only bounded away from zero when $t$ is below the threshold given, but tends to one.  On the other hand, the proof based on Fano's inequality extends more easily to noisy settings.  Extending the result to adaptive test designs (e.g., again see \cite{Bal13}) is also straightforward, but in this paper we focus exclusively on non-adaptive designs.
% \end{remark}

\section{Algorithmic Upper Bounds} \label{sec:algs}

\subsection{COMP Algorithm} \label{sec:comp}

Adopting the terminology from the group testing literature, the COMP algorithm is described in Algorithm \ref{alg:comp}.  The simple idea is that if two nodes appear in a negative test, then the corresponding edge must be absent from $G$.  Hence, all such edges are ruled out, and the remaining edges are declared to be present.  Once Lemma \ref{lem:typical} is in place, the theoretical analysis of COMP becomes very simple and similar to standard group testing \cite{Cha11}, leading to the following.

\begin{algorithm}[!t]
    \caption{Combinatorial Orthogonal Matching Pursuit (COMP)} \label{alg:comp}
    \begin{algorithmic}[1]
        \Require Test designs $\{\Lc^{(i)}\}_{i=1}^t$, outcomes $\Yv=(Y^{(1)},\dotsc,Y^{(t)})$
        \State Initialize $\Ehat$ to contain all ${n \choose 2}$ edges
        \For {{\bf each} $i$ such that $Y^{(i)} = 0$} 
            \State Remove all edges from $\Ehat$ whose nodes are both in $\Lc^{(i)}$
        \EndFor
        \State \Return $\Ghat = (V,\Ehat)$
    \end{algorithmic}
\end{algorithm}

\begin{theorem} \label{thm:comp}
    Under the setup of Section \ref{sec:setup} with $q = \Theta(n^{2(\theta-1)})$ for some $\theta \in (0,1)$, and Bernoulli testing with parameter $\nu = 1$, the COMP algorithm achieves $\pe \to 0$ as long as
    \begin{equation}
        t \ge \big( 2e \cdot \kbar \log n \big) (1+\eta) \label{eq:comp}
    \end{equation}
    for arbitrarily small $\eta > 0$.
\end{theorem}
\begin{proof}
    The graph properties given in the definition \eqref{eq:typical} of $\Tc(\epsilon_n)$ facilitate a direct analysis of the probability that the two nodes of a given non-edge fail to be included together in any negative test, and a union bound over all non-edges establishes the claim.  See Appendix \ref{app:pf_comp} for details.
\end{proof}

\subsection{DD Algorithm} \label{sec:dd}

Since we work on the assumption that edges are rare (i.e., $q \ll 1$), one would expect that COMP's approach of assuming edges are present (unless immediately proven otherwise) can be highly suboptimal.  The DD algorithm,\footnote{For the graph learning problem, one may prefer to name the algorithm {\em Definite Edges}, but we prefer to maintain consistency with the group testing literature \cite{Ald14a}.} described in Algorithm \ref{alg:dd}, overcomes this limitation by assuming edges are absent unless immediately proven otherwise.  The way to prove the presence of the edge is to use COMP to rule out non-edges, mark the remaining pairs as {\em possible edges} ($\PE$), and then look for positive tests containing only a single pair from $\PE$.  The analysis of DD is a fair bit more challenging than COMP, but gives an improved bound, as stated in the following.

\begin{algorithm}[!t]
    \caption{Definite Defectives (DD)} \label{alg:dd}
    \begin{algorithmic}[1]
        \Require Test designs $\{\Lc^{(i)}\}_{i=1}^t$, outcomes $\Yv=(Y^{(1)},\dotsc,Y^{(t)})$
        \State Initialize $\Ehat = \emptyset$, and initialize $\PE$ to contain all ${n \choose 2}$ edges
        \For {{\bf each} $i$ such that $Y^{(i)} = 0$} 
            \State Remove all edges from $\PE$ whose nodes are both in $\Lc^{(i)}$
        \EndFor
        \For {{\bf each} $i$ such that $Y^{(i)} =  1$} 
            \State If the nodes from $\Lc^{(i)}$ cover exactly one edge in $\PE$, add that edge to $\Ehat$
        \EndFor
        \State \Return $\Ghat = (V,\Ehat)$
    \end{algorithmic}
\end{algorithm}

\begin{theorem} \label{thm:dd}
    Under the setup of Section \ref{sec:setup} with $q = \Theta(n^{2(\theta-1)})$ for some $\theta \in (0,1)$, and Bernoulli testing with parameter $\nu = 1$, the DD algorithm achieves $\pe \to 0$ as long as
    \begin{equation}
        t \ge \big( 2 \max\{\theta, 1-\theta\}e \cdot \kbar \log n \big) (1+\eta) \label{eq:dd}
    \end{equation}
    for arbitrarily small $\eta > 0$.
\end{theorem}
\begin{proof}
    The proof is based on analyzing the two steps separately.  In the first step, we show that with high probability not too many non-edges are included in $\PE$, and in the second step, we show that conditioned on this success event in the first step, each true edge is the unique PE in some test with high probability.   The details are given in Appendix \ref{sec:pf_dd}, and the main differences to the standard group testing analysis \cite{Ald14a,Sca18b} are highlighted in Appendix \ref{sec:differences}.
\end{proof}

\section{SSS Algorithm Lower Bound} \label{sec:sss}

It is a standard result that under any random graph model, the optimal decoder (in the sense of minimizing $\pe = \PP[\Ghat \ne G]$) is the one that declares $\Ghat$ to be the most probable graph that would have produced the observation vector $\Yv = (Y^{(1)},\dotsc,Y^{(t)})$ if it were the true graph.  Under the Erd\H{o}s-R\'enyi graph model, graphs with fewer edges are always more likely, so this decoder simply searches for the graph with the fewest edges that is {\em satisfying} in the sense of being consistent with $\Yv$.  This leads to the SSS algorithm described in Algorithm \ref{alg:sss}.  Similarly to \cite{Mal12}, this algorithm amounts to an integer program, which may be hard to solve efficiently in general.

Despite this computational challenge, a key utility of studying SSS is as follows.  Since it is the optimal decoding algorithm, a lower bound on the number of tests it requires is also a lower bound for {\em any} decoding algorithm.  In the following theorem, we provide such a lower bound with respect to random Bernoulli test designs.  While such a lower bound is, in a sense, weaker than that of Theorem \ref{thm:conv} (because that result holds for arbitrary test designs), it leads to the important conclusion that one cannot hope to improve on the bound for DD for $\theta > \frac{1}{2}$ unless one moves beyond Bernoulli test designs.  See Figure \ref{fig:rates} for an illustration.

\begin{algorithm}[!t]
    \caption{Smallest Satisfying Set (SSS)} \label{alg:sss}
    \begin{algorithmic}[1]
        \Require Test designs $\{\Lc^{(i)}\}_{i=1}^t$, outcomes $\Yv=(Y^{(1)},\dotsc,Y^{(t)})$
        \State Find $\Ehat$ that minimizes $|\Ehat|$ subject to $\phi_{\Ehat}(\Lc^{(i)}) = Y^{(i)}$ for all $i=1,\dotsc,t$, where the function $\phi_E(\Lc) = \vee_{(i,j) \in E} \{ \{i,j\} \subseteq \Lc \}$ corresponds to the observation model \eqref{eq:gt_noiseless_model}.
        \State \Return $\Ghat = (V,\Ehat)$ 
    \end{algorithmic}
\end{algorithm}

\begin{theorem} \label{thm:sss}
    Under the setup of Section \ref{sec:setup} with $q = \Theta(n^{2(\theta-1)})$ for some $\theta \in (0,1)$, and Bernoulli testing with an arbitrary choice of $\nu > 0$, the SSS algorithm yields $\pe \to 1$ whenever
    \begin{equation}
        t \le \big( 2\theta e \cdot \kbar \log n \big) (1-\eta) \label{eq:sss}
    \end{equation}
    for arbitrarily small $\eta > 0$.
\end{theorem}
\begin{proof}
    The proof is based on the fact that if an edge is {\em masked} (i.e., its nodes never appear together in any test without those of a different edge), then removing that edge from $E$ will produce a smaller satisfying set, meaning that the algorithm fails to output $E$.  The details are given in Appendix \ref{sec:pf_sss}, and the main differences to the standard group testing analysis \cite{Ald14a} are highlighted in Appendix \ref{sec:differences}.
\end{proof}

\section{Sublinear-Time Decoding} \label{sec:sublinear}

A standard implementation of COMP or DD yields decoding complexity $O(n^2 t)$, which may be infeasible when $n$ is large and decoding time is limited.  To attain {\em sublinear-time} decoding, considerably different algorithms are needed, as one certainly cannot rely on marking non-edges one by one.  In Algorithm \ref{alg:groteque}, we informally outline a sublinear-time decoding algorithm that builds on the ideas of the GROTESQUE algorithm for group testing \cite{Cai13}.  We find it most convenient to formally describe the key components while simultaneously performing the theoretical analysis; see Sections \ref{sec:mult_test} and \ref{sec:loc_test}.  For the purpose of understanding the algorithm, it suffices to note the following:
\begin{itemize}[leftmargin=5ex,itemsep=0ex,topsep=0.25ex]
    \item A {\em multiplicity test} performs a number $\tmul$ of group tests in which the items from a given bundle are included independently with probability $\frac{1}{\sqrt 2}$. By counting the number of positive tests, one can determine with high probability whether or not the bundle covers exactly one edge.
    \item A {\em location test} performs a sequence of carefully-designed tests that permit the identification of the unique edge in a given bundle, provided that bundle indeed only covers one edge. 
\end{itemize}
The resulting number of tests and runtime are given in the following theorem.  Note that in contrast to the previous sections, here our focus is on the scaling laws and not the implied constants.  This is due to the fact that attaining 
sharp constant factors with sublinear-time decoding has remained an open challenge even in the simpler group testing setting \cite{Cai13,Lee15a,Ina19,Bon19a}.

\begin{algorithm}[!t]
    \caption{Group Testing Quick and Efficient (GROTESQUE) -- Informal Outline} \label{alg:groteque}
    \begin{algorithmic}[1]
        \Require Number of bundles $B$, inclusion probability $r$
        \State Form bundles $\Bc_1,\dotsc,\Bc_B$ by independently including each node in each $\Bc_b$ with probability $r$
        \State Initialize $\Ehat = \emptyset$
        \For {{\bf each} $b=1,\dotsc,B$}
            \State Perform a multiplicity test ({\em cf.}, Section \ref{sec:mult_test}) on $\Bc_b$
            \If{ multiplicity test returned ``single edge''}
                \State Perform location test ({\em cf.}, Section \ref{sec:loc_test}) on $\Bc_b$ and add the resulting edge to $\Ehat$
            \EndIf
        \EndFor
        \State \Return $\Ghat = (V,\Ehat)$
    \end{algorithmic}
\end{algorithm}

\begin{theorem} \label{thm:grotesque}
    Under the setup of Section \ref{sec:setup} with $q = \Theta(n^{2(\theta-1)})$ for some $\theta \in (0,1)$, the GROTESQUE test design and decoding algorithm achieves $\pe \to 0$ with $t = O( \kbar \cdot \log \kbar \cdot \log^2 n )$ tests, and the decoding time behaves as $O( \kbar \log^2 \kbar + \kbar \log n )$ with probability approaching one.
\end{theorem}

The proof is given below after a short discussion.
While it may seem unusual to have a decoding time smaller than the number of tests, this is because the decoder is allowed to selectively decide which tests to make use of, and does not end up using them all.  (We implicitly assume that fetching the result of a given test can be done in constant time.) Comparing to Theorems \ref{thm:comp} and \ref{thm:dd}, we see that the number of tests performed has increased by a $\log \kbar \cdot \log n$ factor.  On the other hand, the decoding time is nearly optimal: An analogous argument to Theorem \ref{thm:conv} reveals an $\Omega(\kbar \log n)$ lower bound, and the upper bound in Theorem \ref{thm:grotesque} matches this result when $\log^2 \kbar = O(\log n)$, and more generally comes within at most a single logarithmic factor.

We briefly mention that the {\em encoding} time (i.e., placing nodes in tests) is certainly not sublinear, so the advantage of sublinear decoding time is most beneficial when the encoding time does not pose a bottleneck (e.g., due to an efficient parallel implementation and/or pre-processing).

\subsection{Proof Step 1 -- Bundles of Tests} \label{sec:bundles_main}

Since the number of edges $k$ and maximal degree $d$ behave as stated in \eqref{eq:typical} for some $\epsilon_n = o(1)$ with probability approaching one (see Lemma \ref{lem:typical}), it suffices to establish the claims of Theorem \ref{thm:grotesque} conditioned on an arbitrary graph $G$ satisfying such properties.  We implicitly condition on such a graph throughout the analysis.

As described in Algorithm \ref{alg:groteque}, we form a number $B$ of ``bundles'' of tests, where each node is placed in each bundle with probability $r \in (0,1)$.  In Appendix \ref{sec:bundles}, we use a direct probabilistic analysis to show that under a choice satisfying $B = \big(4\kbar \log \kbar\big) (1+o(1))$, we have with probability $1-o(1)$ that every edge is the unique one in at least one bundle.

\subsection{Proof Step 2 -- Multiplicity Tests} \label{sec:mult_test}

We perform a multiplicity test on each bundle by performing a series of (group) tests in which every node is independently included with probability $\frac{1}{\sqrt 2}$.  For each such test:
\begin{itemize}[leftmargin=5ex,itemsep=0ex,topsep=0.25ex]
    \item If there are no edges, the output is always $0$;
    \item If there is exactly one edge, each output equals $1$ with probability $\frac{1}{2}$;
    \item If there are multiple edges, each output equals $1$ with probability strictly higher than $\frac{1}{2}$.  To see this, first observe that if there are two edges $e_1,e_2$ among $3$ nodes, then the probability of a positive test is 
    % \begin{equation}
        $\PP[e_1 \cup e_2] = \PP[e_1] + \PP[e_2] - \PP[e_1 \cup e_2] = \frac{1}{2} + \frac{1}{2} - \frac{1}{2\sqrt{2}} > 0.646,$
    % \end{equation}
    whereas if there are two disjoint edges $e_1,e_2$ then a similar calculation yields $\PP[e_1 \cup e_2] \ge \frac{3}{4}$.  Hence, the overall probability of a positive test is at least $0.646$.
\end{itemize}
Based on these observations, we declare each bundle to have a single edge if and only if the proportion of $1$'s lies in $\big(0,0.573\big)$.  Trivially, if the number of edges is zero, we never make a mistake.  On the other hand, if the number of edges is one or more than one, we can apply Hoeffding's inequality with a margin of at least $0.073$; hence, using $\tmul$ tests we have
% \begin{equation}
    $\PP[{\rm misclassification}] \le 2e^{-2 \tmul \cdot 0.073^2 }$.
% \end{equation}
Taking the union bound over the $B$ bundles, and noting that $2 \times 0.073^2 > 0.01$, we find that we can classify all of the bundles correctly with probability approaching one when $\tmul = (100 \log B) (1+o(1))$, so that the total number of tests used is $\tmul B = (100 B \log B) (1+o(1))$.

\subsection{Proof Step 3 -- Location Tests} \label{sec:loc_test}

For location testing, we assign each node a unique binary string of length $L = \lceil \log_2 n \rceil$.  In the following, we consider an arbitrary bundle containing a single edge.  For ease of presentation, we describe the location test as though the algorithm could adaptively perform tests, and then we describe how the same can be done non-adaptively.

{\bf Adaptive location test.} The following procedure constructs two binary strings $A$ and $B$ of length $L$; these strings will index the two nodes in the bundle that have an edge between them.  For each $\ell=1,\dotsc,L$, we do the following:
\begin{enumerate}[leftmargin=5ex,itemsep=0ex,topsep=0.25ex]
    \item Test all nodes with a $0$ in their $\ell$-th bit.  If the test is positive, label both $A_{\ell}$ and $B_{\ell}$ as zero.
    \item Test all nodes with a $1$ in their $\ell$-th bit.  If the test is positive, label both $A_{\ell}$ and $B_{\ell}$ as one.
    \item If neither of the preceding tests is positive, we know that $A_{\ell} \ne B_{\ell}$, so we do the following:
    \begin{enumerate}[leftmargin=5ex,itemsep=0ex,topsep=0.25ex]
        \item If this is the first $\ell$ for which this case is encountered, then assign $A_{\ell} = 0$ and $B_{\ell} = 1$ (the other way around would be equally valid).
        \item Otherwise, let $\ell' < \ell$ be an index where we encountered this case, and do the following:
        \begin{enumerate}[leftmargin=5ex,itemsep=0ex,topsep=0.25ex]
            \item Let $v \in \{0,1\}$ be the bit value that was assigned to $A_{\ell'}$.
            \item Perform a test containing all nodes whose bit strings $(v_1,\dotsc,v_L)$ have $v_{\ell} = v_{\ell'} = v$ and also all those that have  $v_{\ell} = v_{\ell'} = 1-v$.  
            \item If the test is positive, then assign $A_{\ell} = v$ and $B_{\ell} = 1-v$.  Otherwise, assign $A_{\ell} = 1-v$ and $B_{\ell} = v$.
        \end{enumerate}
    \end{enumerate}
\end{enumerate}
The idea of step 3(b)(ii) is that we already know that the edges corresponding to $A$ and $B$ must have $v$ and $1-v$ in bit position $\ell'$ respectively, so we perform the test described to check whether the same is true of position $\ell$.  If it is not true, then with $A$ and $B$ differing in their $\ell$-th bit, the only remaining case is that $A$ and $B$ have $1-v$ and $v$ in bit position $\ell$ respectively.

{\bf Non-adaptive location test.}  The only types of tests that the adaptive algorithm above uses are (i) those used in Steps 1 and 2 with all nodes having a given $\ell$-th bit value; and (ii) those used in Step 3b containing all nodes with some $v_{\ell} = v_{\ell'} = v$ and some other $v_{\ell} = v_{\ell'} = 1-v$.  There are only $2L$ possible such tests of type (i), and $2{L \choose 2}$ possible tests of type (ii).  Hence, we can perform the tests non-adaptively by taking all such possible tests in advance.  Moreover, we don't have to look at all their outcomes, but rather only those that we would have taken in the adaptive setting.

Hence, the number of group tests per location test is at most $2\lceil \log_2 n \rceil + 2 \cdot \frac{1}{2} \lceil \log_2 n \rceil^2 = (\log_2 n)^2 (1+o(1))$, so if we perform one location test for each bundle (and again only actually use those that we need) then the total is $B(\log_2 n)^2 (1+o(1))$.

\subsection{Proof Step 4 -- Total Number of Tests and Decoding Time}

The claims on the number of tests and runtime stated in Theorem \ref{thm:grotesque} follow easily from the above analysis, and the details are deferred to Appendix \ref{sec:bundles}.

\section{Numerical Experiments}

We complement our theoretical findings with numerical experiments comparing COMP, DD, SSS, and a linear programming (LP) relaxation of SSS (analogous to \cite{Mal12}).\footnote{The code is available at \url{https://github.com/scarlett-nus/er_edge_det}.}
% \footnote{This relaxation is done in the same way as the standard group testing problem \cite{Mal12}, relaxing the integer $\{0,1\}$ constraint to the interval $[0,1]$ and then performing rounding.  This significantly reduces the computation.}   
Figure \ref{fig:numerical1} shows the success probability as a function of the number of tests in two cases:
(i) $n=50$ and $\kbar = 5$; (ii) $n = 200$ and $\kbar = 200$.  In each case, we set $\nu = 1$ and compute the error probability averaged over $2000$ trials.
In the first case, we observe that the SSS and LP curves are very close, and require the fewest tests; DD requires more tests, and COMP requires the most.  In the second case, we omit SSS due to its computational complexity, but we observe a similar ordering between LP, DD, and COMP.  In both cases, the relative performance between the algorithms is consistent with our theoretical findings:
\begin{itemize}[leftmargin=5ex,itemsep=0ex,topsep=0.25ex]
    \item The first case is a sparse setting, and the performance curves for COMP and DD are relatively closer, which is consistent with the fact that COMP and DD achieve the same theoretical bound in the sparse limit $\theta \to 0$ (see Figure \ref{fig:rates}).
    \item The second case is a denser setting, and the gap between LP and DD is narrower, which is consistent with the fact that the theoretical bounds for DD and SSS coincide in denser regimes.
\end{itemize}
In Appendix \ref{app:more_exp}, we provide similar plots for varying choices of $(\kbar,n)$ in order to demonstrate that the dependence of the number of tests on $\kbar$ and $n$ is in general agreement with our theory.

% In the second case, the information-theoretic threshold is $t \gtrsim 1327$ and the COMP threshold is  $t \lesssim 5761$.  Here we omit SSS due to its computational complexity.

% when $n = 100$ and $\nu = 1$ in two cases: (i) a sparser regime $\kbar = 10$; (ii) a denser regime $\kbar = 200$.  Note that  ignoring higher-order asymptotic terms, the information-theoretic threshold in Theorem \ref{thm:conv} evaluates to $t \gtrsim 89$ and $t \gtrsim 925$ for the two cases respectively, and the threshold for COMP evaluates to $t \lesssim 250$ and $t \lesssim 5007$ respectively.  

% This example is consistent with our theoretical findings, with the LP curve being slightly to the right of the information-theoretic lower bound, the number of tests for DD being slightly higher, and COMP requiring the most tests.  The gap between LP and DD becomes narrow in denser scenarios, which is consistent with the fact that the theoretical bounds for DD and SSS coincide in denser regimes (see Figure \ref{fig:rates}).

\begin{figure}
    \begin{centering}
        {% This file was created by matlab2tikz.
\definecolor{dark green}{rgb}{0,0.7,0}%
\begin{tikzpicture}

\begin{axis}[%
width=5.4cm,
height=4cm,
scale only axis,
xmin=25,
xmax=450,
xtick={ 50, 100, 150, 200, 250, 300, 350, 400, 450},
xlabel style={font=\color{white!15!black}},
xlabel={Number of tests},
ymin=0,
ymax=1,
ytick={  0, 0.2, 0.4, 0.6, 0.8,   1},
ylabel style={font=\color{white!15!black}},
ylabel={Success probability},
axis background/.style={fill=white},
%axis x line*=bottom,
%axis y line*=left,
xmajorgrids,
ymajorgrids,
legend style={at={(0.99,0.015)}, anchor=south east, legend cell align=left, align=left, draw=white!15!black, font= \scriptsize,inner xsep=1pt, inner ysep=1pt}
]
\addplot [color=black, line width=1.0pt, mark=o, mark options={solid, black}]
  table[row sep=crcr]{%
25	0.083\\
50	0.3405\\
75	0.589\\
100	0.7455\\
125	0.829\\
150	0.8865\\
175	0.933\\
200	0.944\\
225	0.964\\
250	0.975\\
275	0.9805\\
300	0.982\\
325	0.992\\
350	0.99\\
375	0.9915\\
400	0.9955\\
425	0.9955\\
450	0.9945\\
};
\addlegendentry{SSS}

\addplot [color=red, line width=1.5pt]
  table[row sep=crcr]{%
25	0.0815\\
50	0.315\\
75	0.561\\
100	0.7265\\
125	0.806\\
150	0.875\\
175	0.924\\
200	0.935\\
225	0.961\\
250	0.972\\
275	0.9775\\
300	0.9805\\
325	0.9895\\
350	0.988\\
375	0.9895\\
400	0.995\\
425	0.995\\
450	0.9955\\
};
\addlegendentry{LP}

\addplot [color=blue, dash pattern=on 2pt off 3pt on 5pt off 3pt, line width=1.5pt]
  table[row sep=crcr]{%
25	0.016\\
50	0.1555\\
75	0.376\\
100	0.559\\
125	0.6795\\
150	0.7695\\
175	0.854\\
200	0.8735\\
225	0.9095\\
250	0.933\\
275	0.9545\\
300	0.9575\\
325	0.975\\
350	0.9795\\
375	0.983\\
400	0.9865\\
425	0.984\\
450	0.992\\
};
\addlegendentry{DD}

\addplot [color= dark green, dash pattern= on 5pt off 4pt, line width=1.5pt]
  table[row sep=crcr]{%
25	0.001\\
50	0.057\\
75	0.187\\
100	0.33\\
125	0.457\\
150	0.565\\
175	0.66\\
200	0.717\\
225	0.76\\
250	0.8145\\
275	0.854\\
300	0.87\\
325	0.8945\\
350	0.9135\\
375	0.93\\
400	0.9405\\
425	0.9415\\
450	0.9585\\
};
\addlegendentry{COMP}

\end{axis}
\end{tikzpicture}%
~~~
\definecolor{dark green}{rgb}{0,0.7,0}%
\begin{tikzpicture}
\pgfplotsset{scaled x ticks=false}
\begin{axis}[%
width=5.1cm,
height=4cm,
scale only axis,
xmin=3000,
xmax=10000,
xtick={4000,  6000, 8000,10000},
minor xtick={3000,  5000, 7000,9000},
xlabel style={font=\color{white!15!black}},
xlabel={Number of tests},
ymin=0,
ymax=1,
ytick={  0, 0.2, 0.4, 0.6, 0.8,   1},
ylabel style={font=\color{white!15!black}},
ylabel={Success probability},
axis background/.style={fill=white},
%axis x line*=bottom,
%axis y line*=left,
xmajorgrids,
xminorgrids,
ymajorgrids,
legend style={at={(0.99,0.015)}, anchor=south east, legend cell align=left, align=left, draw=white!15!black, font=\scriptsize,inner xsep=1pt, inner ysep=1pt}
]
\addplot [color=red, line width=1.5pt]
  table[row sep=crcr]{%
3000	0.0285\\
3250	0.1\\
3500	0.235\\
3750	0.372\\
4000	0.518\\
4250	0.6555\\
4500	0.759\\
4750	0.8185\\
5000	0.878\\
5250	0.915\\
5500	0.941\\
5750	0.955\\
6000	0.9755\\
6250	0.983\\
6500	0.9865\\
6750	0.99\\
7000	0.9935\\
7250	0.9965\\
7500	0.9985\\
7750	0.9985\\
8000	0.998\\
8250	0.999\\
8500	0.9995\\
8750	1\\
9000	0.999\\
9250	1\\
9500	0.9995\\
9750	1\\
10000	1\\
};
\addlegendentry{LP}

\addplot [color=blue, dash pattern=on 2pt off 3pt on 5pt off 3pt, line width=1.5pt]
  table[row sep=crcr]{%
3000	0\\
3250	0.006\\
3500	0.0525\\
3750	0.1695\\
4000	0.304\\
4250	0.4855\\
4500	0.6405\\
4750	0.7665\\
5000	0.8435\\
5250	0.905\\
5500	0.9295\\
5750	0.946\\
6000	0.968\\
6250	0.9755\\
6500	0.9765\\
6750	0.993\\
7000	0.9925\\
7250	0.9935\\
7500	0.9965\\
7750	0.9975\\
8000	0.999\\
8250	0.9985\\
8500	0.9995\\
8750	0.998\\
9000	1\\
9250	0.9985\\
9500	0.999\\
9750	0.9995\\
10000	1\\
};
\addlegendentry{DD}

\addplot [color=dark green, dash pattern= on 5pt off 4pt, line width=1.5pt]
  table[row sep=crcr]{%
3000	0\\
3250	0\\
3500	0\\
3750	0\\
4000	0\\
4250	0\\
4500	0\\
4750	0.001\\
5000	0.007\\
5250	0.0215\\
5500	0.042\\
5750	0.1065\\
6000	0.1435\\
6250	0.2495\\
6500	0.337\\
6750	0.4705\\
7000	0.5565\\
7250	0.6415\\
7500	0.712\\
7750	0.7675\\
8000	0.8475\\
8250	0.8855\\
8500	0.9165\\
8750	0.9275\\
9000	0.954\\
9250	0.9635\\
9500	0.9675\\
9750	0.9815\\
10000	0.9825\\
};
\addlegendentry{COMP}

\end{axis}

\end{tikzpicture}%
}
        \par
    \end{centering}
    
    \caption{Performance of the COMP, DD, LP, and SSS algorithms for noiseless group testing under Bernoulli testing with $\nu = 1$, and with $n = 50$ and $\kbar = 5$ (Left); $n = 200$ and $\kbar = 200$ (Right). \label{fig:numerical1}}
\end{figure}
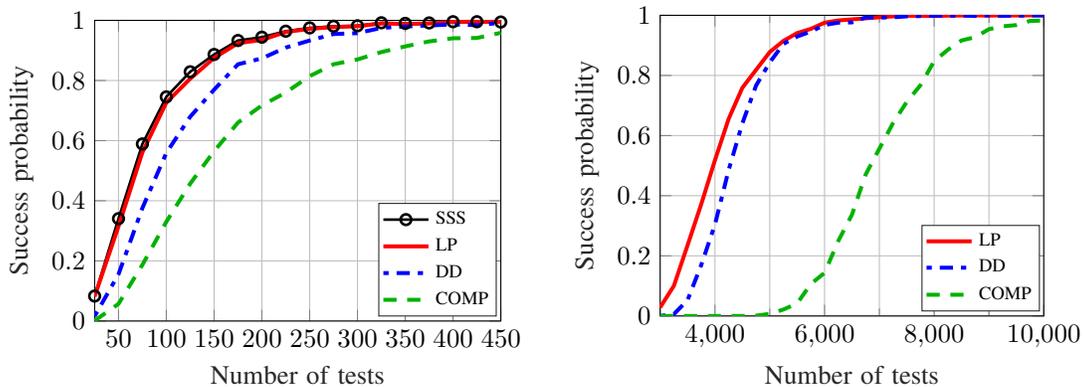

\section{Conclusion}

We have studied the problem of learning Erd\H{o}s-R\'enyi random graphs via edge detecting queries, and demonstrated significantly improved scaling of $O(\kbar \log n)$ compared to worst-case graphs with $k$ edges.  We provided order-optimal bounds for the COMP, DD, and SSS algorithms with explicit constants, showed DD to be optimal under Bernoulli testing when the graph is sufficiently dense ($\theta \ge \frac{1}{2}$), and introduced a sublinear-time algorithm that succeeds with $O(\kbar \log^2 \kbar + \kbar \log n)$ runtime. 
% In addition, we introduced a sublinear-time algorithm attaining vanishing error probability with $O(\kbar \log^2 \kbar + \kbar \log n)$ runtime under a suitably-chosen test design.

Given that the ideas of this paper build on a variety of techniques for small-error group testing, it is natural to pursue further research in directions that were done previously in that setting, including separate decoding of items (edges) \cite{Mal80,Sca17b}, information-theoretic achievability bounds \cite{Sca15b,Coj19}, and near-constant tests-per-item (tests-per-node) designs \cite{Joh16,Coj19}.  Generalizations of our techniques to hypergraph learning \cite{Ang06,Ang08,Dya16,Aba18a} would also be of significant interest.

% {\bf Acknowledgment.} This work was supported by an NUS Early Career Research Award.

\appendices

%{\centering
%    {\huge \bf Supplementary Material}
%    
%    {\large \bf Learning Erd\H{o}s-R\'enyi Random Graphs via Edge Detecting Queries \\ [2mm] {\normalsize Zihan Li, Matthias Fresacher, and Jonathan Scarlett (NeurIPS 2019)} \par }  
%}

% \smallskip
% \bigskip
% % All citations below are to the reference list in the main body.

\section{Proof of Lemma \ref{lem:typical} (High Probability Typicality)} \label{sec:pf_typical}

To prove that there exists a sequence $\epsilon_n \to 0$ such that $\PP[G \in \Tc(\epsilon_n)] \to 1$ as $n \to \infty$, it suffices to show that $\PP[G \in \Tc(\epsilon)] \to 1$ as $n \to \infty$ for arbitrarily small (but fixed) $\epsilon > 0$.

Since $q = \Theta(n^{-2(1-\theta)})$ for some $\theta \in (0,1)$, the average number of edges $\kbar = q{n \choose 2}$ must grow unbounded as $n \to \infty$.  Hence, the fact that $\PP[(1-\epsilon) \kbar \le k \le (1+\epsilon)\kbar] \to 1$ follows via basic binomial concentration (e.g., Chernoff bound).  The bound on the degree is also easy to handle:
\begin{itemize}[leftmargin=5ex,itemsep=0ex,topsep=0.25ex]
    \item For $\theta > \frac{1}{2}$, the per-node degree follows a binomial distribution with mean $(n-1)q = \Theta( n^c )$ for some $c > 0$.  Hence, by the multiplicative form of the Chernoff bound \cite[Sec.~4.1]{Mot10}, the probability of the degree exceeding $2nq$ (which exceeds double the average) behaves as $e^{- \Omega(n^c)}$.  By a union bound over the $n$ nodes, we find that the probability of any node's degree exceeding $2nq$ is at most $n e^{- \Omega(n^c)} \to 0$.
    \item For $\theta \le \frac{1}{2}$, it suffices to consider the case $\theta = \frac{1}{2}$, for which the probability of the maximal degree exceeding $\log n$ is clearly highest.  In this case, a given node's degree follows a binomial distribution with $n-1$ trials and success probability $\Theta\big( \frac{1}{n} \big)$, so the mean behaves as $\Theta(1)$.  Therefore, the probability of the degree exceeding $\log n$ equals the probability of being at least a factor $1 + \Delta$ higher than the mean $(n-1)q$, for some $\Delta = \Theta(\log n)$.  By a standard form of the Chernoff bound \cite[Sec.~4.1]{Mot10}, this occurs with probability at most $e^{- (n-1)q [ (1+\Delta)\log(1+\Delta) - \Delta]}$, which behaves as $e^{- \Omega(\log n \cdot \log \log n)}$ since $(n-1)q = \Theta(1)$.  By a union bound over the $n$ nodes, the probability of any degree exceeding $\log n$ is at most $n e^{- \Omega( \log n \cdot \log \log n)} \to 0$.
\end{itemize} 

By far the most challenging event to handle in the typical set is the final one, $(1-\epsilon)(1-e^{-\nu}) \le \PP_G[Y=1] \le (1+\epsilon)(1-e^{-\nu})$.  The intuition behind the analysis is as follows:
\begin{itemize}[leftmargin=5ex,itemsep=0ex,topsep=0.25ex]
    \item In generic notation, let $A_1,\dotsc,A_N$ be independent events each occurring with probability $r$.  Then $\PP\big[ \bigcup_{i=1}^N A_i  \big] = 1 - (1-r)^N$, which behaves as $(1-e^{-Nr})(1+o(1))$ as $N \to \infty$ with $r = \Theta\big(\frac{1}{N}\big)$.
    \item Letting $S_j = \sum_{1 \le i_1 < \dotsc < i_j \le N} \PP[ A_{i_1} \cap \dotsc \cap A_{i_j} ]$, we know from the Bonferroni inequalities \cite{Gal96} that
    \begin{equation}
        \PP\bigg[ \bigcup_{i=1}^N A_i  \bigg] \le \sum_{j=1}^{j_{\max}} (-1)^{j+1} S_j
    \end{equation}
    for odd $j_{\max}$, and the inequality is flipped for even $j_{\max}$.  
    \item In the special case of independent events and $N \to \infty$ with $r = \Theta\big(\frac{1}{N}\big)$, assuming that $j_{\max} = O(1)$, we have
    \begin{align}
        \sum_{j=1}^{j_{\max}} (-1)^{j+1} S_j
            &=  \sum_{j=1}^{j_{\max}} (-1)^{j+1} {N \choose j} r^j \\
            &= -\sum_{j=1}^{j_{\max}} \frac{1}{j!} (-Nr)^j (1+o(1)), \label{eq:intuition_end}
    \end{align}
    since ${N \choose j} = \frac{N!}{j!(N-j)!} = \frac{1}{j!} N^j (1+o(1))$.  Due to the limit $\sum_{j=1}^{\infty} \frac{1}{j!} (-Nr)^j = e^{-Nr}-1$, \eqref{eq:intuition_end} is arbitrarily close to $1-e^{-Nr}$ for large $j_{\max}$, regardless of whether $j_{\max}$ is even or odd.  As one should expect, this matches the probability computed directly in the first dot point above.
    \item For the graph learning problem, we do not have independent events, but we will still show similar behavior to \eqref{eq:intuition_end} to deduce the precise probability of a positive test given $G$.
\end{itemize}
We now proceed with the formal argument.

{\bf High-probability counting event.} Fix an arbitrary integer $j_{\max} > 0$, and for each $j=1,\dotsc,j_{\max}$ and $\ell = 1,\dotsc,2j-1$, let $U_{j,\ell}(G)$ be the number of sets of exactly $j$ edges in $G$ that collectively consist of exactly $\ell$ nodes.\footnote{For some combinations of $j$ and $\ell$ we trivially have $U_{j,\ell}(G) = 0$, but there is no need for us to explicitly account for this.}  We define the following typicality event for a graph $G$:
\begin{equation}
U_{j,\ell}(G) \le \frac{1}{\delta} {n \choose \ell} { {\ell \choose 2} \choose j } q^j, \quad \forall j=1,\dotsc,j_{\max}, ~~~ \ell = 1,\dotsc,2j-1. \label{eq:Ujl}
\end{equation}
We claim that the probability (with respect to $G \sim \ER(n,q)$) of this occurring is at least $1 - O(\delta)$, where the implied constant depends on $j_{\max}$.

To see this, first note that any fixed collection of $j$ node pairs are all edges with probability $q^j$.  The number of such collections satisfying the condition defining $U_{j,\ell}(G)$ is at most ${n \choose \ell}{ {\ell \choose 2} \choose j }$, i.e., first choose $\ell$ nodes from $n$, and then choose $j$ edges from the corresponding $\ell \choose 2$ possibilities.  Hence,
\begin{equation}
\EE[ U_{j,l}(G) ] \le {n \choose \ell}{ {\ell \choose 2} \choose j } q^j,
\end{equation}
and the $1-O(\delta)$ probability claim follows from Markov's inequality, a union bound over $j=1,\dotsc,j_{\max}$ and $\ell=1,\dotsc,2j-1$, and the assumption that $j_{\max}$ is finite.

{\bf Analysis of $\PP_G[Y=1]$ with respect to a random test.} Let $G = (V,E)$ be a fixed graph whose number of edges $k$ and maximal degree $d$ satisfies the typicality bounds in \eqref{eq:typical}, and that also satisfies the high-probability counting event \eqref{eq:Ujl}.  We write
\begin{equation}
    \PP[Y = 1] = \PP\bigg[ \bigcup_{e \in E} A_e \bigg],
\end{equation}
where $A_e$ is the event that both nodes from $e$ are in the test (and hence $\PP[A_e] = p^2$), and here and subsequently we implicitly condition on $G$ being the graph (i.e., we write $\PP[\cdot]$ in place of $\PP_G[\cdot]$).  The Bonferroni inequality therefore states that
\begin{equation}
    \PP[Y = 1]  \le \sum_{j=1}^{j_{\max}} (-1)^{j+1} S_j
\end{equation}
for odd $j$ (or the reverse for even $j$), where 
\begin{equation}
    S_j = \sum_{ 1 \le i_1 < \dotsc < i_j \le k } \PP\big[ A_{e(i_1)} \cap \dotsc \cap A_{e(i_j)}  \big], \label{eq:Sj}
\end{equation}
and where $A_{e(i)}$ is the $i$-th edge for some arbitrary but fixed ordering of the $k$ edges.

We proceed by characterizing $S_j$ for fixed $j$, assuming $j_{\max} = O(1)$ throughout. We will show that the summation in \eqref{eq:Sj} is asymptotically equivalent to the restricted summation in which the $j$ edges are disjoint, i.e., share no nodes in common.  First note that for disjoint $e(i_1),\dotsc,e(i_j)$, we have $\PP\big[ A_{e(i_1)} \cap \dotsc \cap A_{e(i_j)}  \big] = p^{2j}$.  Since there are trivially at most ${k \choose j}$ ways of choosing $j$ disjoint edges, it follows that
\begin{align}
    \sum_{ \substack{1 \le i_1 < \dotsc < i_j \le k \\ \text{edges disjoint} } } \PP\big[ A_{e(i_1)} \cap \dotsc \cap A_{e(i_j)}  \big] 
        & \le {k \choose j} \cdot p^{2j} \\
        &= \Big( \frac{1}{j!} \cdot k^j \cdot p^{2j} \Big)(1+o(1)), \label{eq:disjoint_ub}
\end{align}
since $\frac{k!}{(k-j)!} = k^j (1+o(1))$ as $k \to \infty$ with $j = O(1)$.

We now seek a matching lower bound to \eqref{eq:disjoint_ub}.  The number of unordered sequences of $j$ disjoint edges is equal to $\frac{1}{j!}$ times the number ordered sequences of $j$ disjoint edges, and to count the latter, we consider a sequential selection of nodes.  Whenever a node is selected, at most $d+1$ nodes are ruled out due to being its neighbor (or itself), so the number of selections is at least $k(k-(d+1))(k-2(d+1))\dotsc (k-(j-1)(d+1))$.  But since $j = O(1)$ and $d = o(k)$ (see \eqref{eq:d_vs_k}), this simply  behaves as $k^j (1+o(1))$, and we deduce that
\begin{equation}
     \sum_{ \substack{1 \le i_1 < \dotsc < i_j \le k \\ \text{edges disjoint} } } \PP\big[ A_{e(i_1)} \cap \dotsc \cap A_{e(i_j)}  \big] \ge \Big( \frac{1}{j!} \cdot k^j \cdot p^{2j} \Big)(1+o(1)). \label{eq:disjoint_lb}
\end{equation}

We now show how to use \eqref{eq:disjoint_ub} and \eqref{eq:disjoint_lb} to deduce upper and lower bounds on $S_j$.  In fact, the lower bound is trivial, since we can simply drop any remaining terms (i.e., those with non-disjoint edges) in \eqref{eq:Sj} by lower bounding the summand by zero.  For the upper bound, however, some additional effort is required.  We let $N(i_1,\dotsc,i_j)$ denote the number of nodes that the edges $e(i_1),\dotsc,e(i_j)$ collectively contain, and write
\begin{equation}
    S_j = \sum_{\ell=1}^{2j} \sum_{ \substack{ 1 \le i_1 < \dotsc < i_j \le k \\ N(i_1,\dotsc,i_k) = \ell }} \PP\big[ A_{e(i_1)} \cap \dotsc \cap A_{e(i_j)}  \big]. \label{eq:Sj2}
\end{equation}
We bound the inner summand separately for each $\ell$.  Note that $\ell = 2j$ corresponds to the disjoint case that we already handled, so we proceed by assuming that $\ell < 2j$.

For each $\ell < 2j$, the summand in \eqref{eq:Sj2} is equal to $p^{\ell}$, and according to the high-probability event in \eqref{eq:Ujl}, the number of such summands is at most $U_{j,\ell} \le \frac{1}{\delta} {n \choose \ell} { {\ell \choose 2} \choose j } q^j$, which we can crudely further upper bound as
\begin{equation}
    U_{j,\ell} \le O\Big( \frac{1}{\delta} \Big) \cdot n^{\ell} q^j,
\end{equation}
since $j$ and $\ell$ are both $O(1)$.  Since $\ell < 2j$, or equivalently $\ell \le 2j-1$ or $j \ge \frac{\ell+1}{2}$, we can write $q^j \le  q^{\ell/2} \cdot q^{1/2}$ (note that $q < 1$), which implies
\begin{equation}
    U_{j,\ell} \le O\Big( \frac{1}{\delta} \Big) \cdot n^{\ell} q^{\ell/2} q^{1/2},
\end{equation}
and hence
\begin{align}
    \sum_{ \substack{ 1 \le i_1 < \dotsc < i_j \le k \\ N(i_1,\dotsc,i_k) = \ell }} \PP\big[ A_{e(i_1)} \cap \dotsc \cap A_{e(i_j)}  \big] 
        &\le O\Big( \frac{1}{\delta} \Big) \cdot n^{\ell} q^{\ell/2} q^{1/2} p^{\ell} \\
        &\le O\Big( \frac{q^{1/2}}{\delta} \Big) \cdot k^{\ell/2} p^{\ell}, \label{eq:ell_simplified}
\end{align}
where we have used $k = \big(\frac{1}{2} qn^2\big)(1+o(1))$.  Now observe that if we choose $\delta$ to tend to zero at a strictly slower rate than $q^{1/2}$, the $O(\cdot)$ term in \eqref{eq:ell_simplified} behaves as $o(1)$.  Since $kp^2 = \Theta(1)$ by our choice of $p$, and $\ell = O(1)$ by assumption, we conclude that overall \eqref{eq:ell_simplified} vanishes as $n \to \infty$, for any $\ell < 2j$.  In contrast, the term \eqref{eq:disjoint_lb} corresponding to $\ell = 2j$ behaves as $\Theta(1)$.  Since there are only a finite number of $\ell$ values, we deduce that the overall outer sum \eqref{eq:Sj2} is asymptotically equivalent to its first term characterized in \eqref{eq:disjoint_ub}--\eqref{eq:disjoint_lb}:
\begin{equation}
    S_j = \Big( \frac{1}{j!} \cdot k^j \cdot p^{2j} \Big)(1+o(1)).
\end{equation}
By identifying this asymptotic expression with \eqref{eq:intuition_end}, with $N = k$ and $r = p^2$, we conclude that the following holds as $j_{\max} \to \infty$:
\begin{equation}
    \PP[Y=1] = \big(1-e^{-kp^2}\big)(1+o(1)) = (1-e^{-\nu}) \cdot (1+o(1)), \label{eq:PY}
\end{equation} 
since $p^2 = \frac{\nu}{\kbar} = \frac{\nu}{k}(1+o(1))$.
In other words, by choosing $j_{\max}$ sufficiently large, we can ensure that under the high-probability event in \eqref{eq:Ujl} and the high-probability bounds on $k$ and $d$ in \eqref{eq:typical}, it holds that $(1-e^{-\nu}) (1-\epsilon) \le \PP[Y=1] \le (1-e^{-\nu}) (1+\epsilon)$ for arbitrarily small $\epsilon > 0$ and sufficiently large $n$.  This completes the proof of Lemma \ref{lem:typical}.

\section{Proof of Theorem \ref{thm:conv} (Converse Bound)} \label{app:pf_converse}

We use a conditional form of Fano's inequality (e.g., \cite[Thm.~3]{Sca19b}) with conditioning on the event that the number of edges $k$ in $G$ satisfies $(1-\epsilon) \kbar \le k \le \kbar (1+\epsilon)$ for small $\epsilon > 0$.  Denoting this event by $\Ac$, and using the usual notation $H(X)$, $H(Y|X)$, $I(X;Y)$, etc.~for entropy and mutual information, Fano's inequality gives
\begin{align}
    \pe &\ge \PP[ \Ac ] \frac{H(G|\Ghat,\Ac={\rm true}) - \log 2 }{ \log |\Gc_{\Ac}| } \\
    &= \PP[ \Ac ] \frac{H(G|\Ac={\rm true}) - I(G; \Ghat|\Ac={\rm true}) - \log 2 }{ \log |\Gc_{\Ac}| },  \label{eq:Fano2}
\end{align}
where $\Gc_{\Ac}$ is the set of graphs such that $(1-\epsilon) \kbar \le k \le \kbar (1+\epsilon)$.

Note that the preceding condition on $k$ is a standard notion of {\em typicality} for collections of independent random variables (in this case, edges).  Using standard properties of typical sets \cite[App.~C]{Sha14}, we have $\PP[\Ac] = 1-o(1)$, $\log |\Gc_{\Ac}|  = {n \choose 2}H_2(q) (1+o(1))$, and $H(G|\Ac={\rm true}) = {n \choose 2}H_2(q) (1+o(1))$, where $H_2(q) = q\log\frac{1}{q} + (1-q)\log\frac{1}{1-q}$ is the binary entropy function.  In addition, the data processing inequality \cite[Sec.~2.8]{Cov01} gives $I(G; \Ghat|\Ac={\rm true}) \le I(G; \Yv |\Ac={\rm true})$, and since $\Yv \in \{0,1\}^t$, this mutual information is further upper bounded by $t \log 2$.  Substituting the preceding findings into \eqref{eq:Fano2} yields
\begin{align}
    \pe \ge \Big( 1 - \frac{t \log 2}{ {n \choose 2}H_2(q) } \Big) (1+o(1)).
\end{align}
Since we consider the regime $q \to 0$, we have $H_2(q) = \big(q\log\frac{1}{q} \big)(1+o(1))$, and hence
\begin{align}
    \pe \ge \Big( 1 - \frac{t \log 2}{ \frac{1}{2} q n^2 \log\frac{1}{q} } \Big) (1+o(1)).
\end{align}
Since $\kbar = \frac{1}{2} q n^2 (1+o(1))$, we conclude that achieving $\pe \to 0$ requires \eqref{eq:t_conv}.

\section{Proof of Theorem \ref{thm:comp} (COMP Upper Bound)} \label{app:pf_comp}

Since the random graph is in the typical set \eqref{eq:typical} with probability approaching one, it suffices to establish that the number of tests \eqref{eq:comp} yields asymptotically vanishing error probability conditioned on an arbitrary typical graph $G \in \Tc_n(\epsilon_n)$, with $\epsilon_n = o(1)$ due to Lemma \ref{lem:typical}.  We implicitly condition on such a graph throughout the analysis.  

Let $(i,j)$ be a given non-edge of $G$.  A particular test fails to identify this non-edge if either (i) $i$ and/or $j$ are not included in the test; or (ii) $i$ and $j$ are both in the test, but there is also an edge covered by the test.  Hence, the probability that a given test fails to identify $(i,j)$ as a non-edge is
\begin{align}
    p_0 := (1-p^2) + p^2 \PP\big[ Y=1 \,\big|\, \{i,j\} \subseteq \Lc\big], \label{eq:p0}
\end{align}
where we recall that $\Lc$ is the set of nodes in the test.
Note that to obtain $Y=1$, we need the test to include either a node with an edge connected to $i$ or $j$, or two separate nodes with an edge between them.  Denoting these two events by $A_1$ and $A_2$, we have
\begin{align}
    \PP\big[ Y=1 \,\big|\, \{i,j\} \subseteq \Lc\big]
        &\le \PP\big[ A_1 \,\big|\, \{i,j\} \subseteq \Lc\big] + \PP\big[ A_2 \,\big|\, \{i,j\} \subseteq \Lc\big] \\
        &\le 2dp + \PP[Y=1],
\end{align}
where the first term follows because there are at most $2d$ nodes connected to $i$ or $j$, and the second term uses the fact that $A_2$ is independent of the event $\{i,j\} \subseteq \Lc$ and in itself implies $Y=1$.  Substituting $\PP[Y=1] = (1-e^{\nu}) (1+o(1))$ in accordance with \eqref{eq:typical}, recalling from \eqref{eq:d_vs_k} that $2dp = o(1)$, and returning to \eqref{eq:p0}, we obtain
\begin{align}
    p_0 &\le 1 - p^2 + p^2 \big( (1-e^{\nu})(1+o(1)) + o(1) ) \\
        &= 1 - p^2 e^{-\nu} (1+o(1)),
\end{align}
since $\nu$ is constant.  Hence, the probability that {\em all} $t$ tests fail to identify $(i,j)$ as a non-edge is
\begin{align}
    p_0^t &= \Big(1 - p^2 e^{-\nu} (1+o(1)) \Big)^t \\
        &\le e^{ -t p^2 e^{-\nu} (1+o(1)) },
\end{align}
since $1-\alpha \le e^{-\alpha}$.
Substituting $p^2 = \frac{\nu}{k} (1+o(1))$ and setting $\nu = 1$, we obtain
\begin{equation}
    p_0^t \le e^{ -\frac{t}{ek} (1+o(1)) },  \label{eq:p0t_ub}
\end{equation}
and by a union bound over at most ${n \choose 2} \le n^2$ non-edges, it follows that
\begin{equation}
    \PP[{\rm error}] \le n^2 e^{ -\frac{t}{ek} (1+o(1)) }.
\end{equation}
Re-arranging, we deduce that $\PP[{\rm error}] \to 0$ as long as
\begin{equation}
    t \ge \big( 2e \cdot k \log n \big) (1+\eta)
\end{equation}
for arbitrarily small $\eta > 0$.  Since $k = \kbar(1+o(1))$ for all typical graphs, and the probability that $G$ is typical tends to one (see Lemma \ref{lem:typical}), we obtain the condition in \eqref{eq:comp}.

\section{Proof of Theorem \ref{thm:dd} (DD Upper Bound)} \label{sec:pf_dd}

Since the random graph is in the typical set \eqref{eq:typical} with probability approaching one, it suffices to establish that the number of tests \eqref{eq:dd} yields vanishing error probability conditioned on an arbitrary typical graph $G \in \Tc_n(\epsilon_n)$, with $\epsilon_n = o(1)$ due to Lemma \ref{lem:typical}.  We implicitly condition on such a graph $G$ throughout the analysis.

\subsection{First Step} \label{sec:dd_first}

The first step of DD gives a set of ``possible edges'' $\PE$ that may contain non-edges.  Let $H_0$ be the total number of non-edges in $\PE$, and let $H_1$ be the number of non-edges in $\PE$ such that at least one of its two nodes forms part of at least one true edge.  Since the total number of non-edges is less than $n^2$, we have from \eqref{eq:p0t_ub} that
\begin{equation}
    \EE[H_0] \le n^2 e^{ -\frac{t}{ek} (1+o(1)) }.  \label{eq:EH0}
\end{equation}
Similarly, since the total number of non-edges sharing a node with a true edge is at most $2kd$ (and also trivially less than $n^2$), we have
\begin{equation}
    \EE[H_1] \le \min\{2kd,n^2\} e^{ -\frac{t}{ek} (1+o(1)) }. \label{eq:EH1}
\end{equation}
By Markov's inequality, it follows for any $\xi_0 > 0$ and $\xi_1 > 0$ that that
\begin{align}
    \PP[H_0 \ge n^{2\xi_0}] &\le n^{2(1 - \xi_0)} e^{ -\frac{t}{ek} (1+o(1)) } \\
    \PP[H_1 \ge n^{2\xi_1}] &\le \min\{2kd,n^2\} n^{-2\xi_1} e^{ -\frac{t}{ek} (1+o(1)) }. \label{eq:H1_Markov}
\end{align}
Re-arranging, we deduce that these two probabilities both vanish as $n \to \infty$ as long as
\begin{gather}
    t \ge \Big( 2(1-\xi_0) e k \log n \Big) (1+\eta), \\
    t \ge (1+\eta) e k \log n \times 
    \begin{cases}
        2(1 - \xi_1) & \frac{3}{4} \le \theta < 1  \\
        4\theta - 1 - 2\xi_1 & \frac{1}{2} < \theta < \frac{3}{4} \\
         2(\theta - \xi_1) & 0 < \theta \le \frac{1}{2}
    \end{cases}
\end{gather}
for arbitrarily small $\eta > 0$; here, the first case uses the $n^2$ term in the $\min\{\cdot\}$ in \eqref{eq:H1_Markov}, the second case uses the $2kd$ term and the fact that $k = \Theta( n^{2\theta} )$ and $d = \Theta(nq) = \Theta( n^{2\theta - 1} )$ for $\theta > \frac{1}{2}$, and the third case uses $k = \Theta( n^{2\theta} )$ and $d = O(\log n)$ for $\theta \le \frac{1}{2}$.

It will shortly prove convenient to ensure that $H_0= o(k)$ and $H_1 = o(\sqrt{k})$ (with high probability).  We achieve this by setting $\xi_0$ to be arbitrarily close to (but still less than) $\theta$, and similarly $\xi_1$ arbitrarily close to $\theta/2$, so that the above requirements simplify to
\begin{gather}
t \ge \big( 2(1-\theta) e k \log n \big) (1+\eta), \label{eq:dd_t1_1} \\
t \ge (1+\eta) e k \log n \times
\begin{cases}
2 - \theta & \frac{3}{4} \le \theta < 1  \\
3\theta - 1 & \frac{1}{2} < \theta < \frac{3}{4} \\
 \theta  & 0 < \theta \le \frac{1}{2}
\end{cases} \label{eq:dd_t1_2} 
\end{gather}
for arbitrarily small $\eta > 0$.

\subsection{Second Step} \label{sec:dd_second}

We condition on the above-mentioned high-probability events from the first step holding: $H_0 = o(k)$ and $H_1 = o(\sqrt{k})$.  In addition, we may assume that the number of positive tests $T_+$ satisfies
\begin{equation}
    T_{+} = t (1-e^{-\nu}) (1+o(1)),
\end{equation}
as this occurs with probability approaching one as $t \to \infty$ in accordance with \eqref{eq:PY} and standard concentration (e.g., Hoeffding's inequality).  We henceforth condition on any such $T_+ = t_+$, as well as a set $\PE = \peset$ that yields $H_0 = h_0 = o(k)$ and $H_1 = h_1 = o(\sqrt{k})$.

For a given true edge $(i,j)$, let $T_{i,j}$ be the number of tests containing $(i,j)$ and no other edges from PE.  We claim that the distribution of $T_{i,j}$ given $t_+$ and $\peset$ is 
\begin{equation}
    (T_{i,j} \,|\, t_+,\peset) \sim \Binomial\Big( t_+, \frac{q_{i,j}}{q_+} \Big), \label{eq:Tij}
\end{equation}
where $q_{i,j}$ is the conditional probability (given $\PE = \peset$) of a given test including $(i,j)$ and no other pairs from $\PE$, and $q_+ = (1-e^{-\nu}) (1+o(1))$ is the unconditional probability of a positive test.  While the distribution \eqref{eq:Tij} is intuitive, its derivation is somewhat tedious, so it is postponed to the end of this appendix (Section \ref{sec:dd_cond_distr}).

We proceed by lower bounding $q_{i,j}$.  For a given random test, let $A_1$ be the event that the test includes a pair in PE connected to either $i$ or $j$, and let $A_2$ be the event that the test includes a pair in PE connected to neither $i$ nor $j$.  Given that $(i,j)$ is in the test (which occurs with probability $p^2$), $(i,j)$ fails to be the unique PE in the test only if either $A_1$ or $A_2$ occurs, so
\begin{align}
    q_{i,j} 
        &= p^2 \cdot \big( 1 - \PP[ A_1 \cup A_2 \,|\, \peset, \{i,j\} \subseteq \Lc ] \big) \label{eq:qij_step1} \\
        &\ge p^2 \cdot \big( 1 - \PP[ A_1  \,|\, \peset, \{i,j\} \subseteq \Lc ] - \PP[ A_2  \,|\, \peset, \{i,j\} \subseteq \Lc ] \big) \\
        &\ge p^2 \cdot \big( 1 - (2d+h_1)p - \PP[ A_2  \,|\, \peset, \{i,j\} \subseteq \Lc ] \big) \label{eq:qij_step3} \\
        &= p^2 \cdot \big( 1 - o(1) - \PP[ A_2  \,|\, \peset, \{i,j\} \subseteq \Lc ] \big), \label{eq:qij_step4}
\end{align}
where the $(2d+h_1)p $ term in \eqref{eq:qij_step3} arises from at most $2d$ true edges connected to $i$ or $j$ and at most an additional $h_1$ non-edges in $\PE$ connected to $i$ or $j$, and \eqref{eq:qij_step4} follows from the fact that $p = \Theta\big( \frac{1}{\sqrt k} \big)$ along with $d = o(\sqrt{k})$ and $h_1 = o(\sqrt{k})$.

To characterize the probability of $A_2$ in \eqref{eq:qij_step4}, we write $A_2 = A'_2 \cup A''_2$, where $A'_2$ is the event that the test includes a true edge connected to neither $i$ nor $j$, and $A''_2$ is the event that the test includes a non-edge in PE connected to neither $i$ nor $j$.  We have
\begin{align}
    \PP[ A_2  \,|\, \peset, \{i,j\} \subseteq \Lc ]  
        &\le \PP[ A'_2  \,|\, \peset, \{i,j\} \subseteq \Lc ]  + \PP[ A''_2  \,|\, \peset, \{i,j\} \subseteq \Lc ] \\
        &=  \PP[ A'_2 ]  + \PP[ A''_2  \,|\, \peset, \{i,j\} \subseteq \Lc ] \label{eq:A2_step2} \\
        &\le \PP[Y = 1] + h_0 p^2  \label{eq:A2_step3} \\
        &= (1-e^{-\nu}) (1+o(1)),  \label{eq:A2_step4}
\end{align}
where \eqref{eq:A2_step2} uses the fact that $A'_2$ is independent of all events being conditioned on (since $A'_2$ concerns only true edges separate from $\{i,j\}$), the first term in \eqref{eq:A2_step3} uses the fact that the event $A'_2$ implies $Y=1$, the second term in \eqref{eq:A2_step3} uses the fact that there are at most $h_0$ possible pairs each included with probability $p^2$, and \eqref{eq:A2_step4} uses $\PP[Y = 1] = (1-e^{-\nu}) (1+o(1))$ along with $h_0 = o(k)$ and $p^2 = \Theta\big(\frac{1}{k}\big)$.

Substituting \eqref{eq:A2_step4} into \eqref{eq:qij_step4} gives
\begin{align}
    q_{i,j} &\ge p^2 e^{-\nu} (1+o(1)) \\
        &= \frac{\nu e^{-\nu}}{k} (1+o(1)), 
\end{align}
recalling that $p^2 = \frac{\nu}{k} (1+o(1))$.  Returning to \eqref{eq:Tij}, we find that
\begin{equation}
T_{i,j} \sim \Binomial\Big( t_+, \frac{\nu}{k} \cdot \frac{e^{-\nu}}{1-e^{-\nu}} \cdot (1+o(1)) \Big), \label{eq:Tij2}
\end{equation}
and since $t_+ = t(1-e^{-\nu})(1+o(1))$ and a $\Binomial(N,r)$ random variable equals zero with probability $(1-r)^N \le e^{-Nr}$, it follows that
\begin{equation}
    \PP[ T_{i,j} = 0 ] \le \exp\bigg( - \frac{t}{k} \cdot \nu e^{-\nu} \cdot (1+o(1)) \bigg),
\end{equation}
and hence
\begin{equation}
\PP\bigg[ \bigcup_{(i,j) \in E} \{ T_{i,j} = 0 \} \bigg] \le k \exp\bigg( - \frac{t}{k} \cdot \nu e^{-\nu} \cdot (1+o(1)) \bigg).
\end{equation}
Re-arranging, setting $\nu = 1$, and writing $\log k = (2\theta \log n)(1+o(1))$, we find that the second step of DD succeeds as long as
\begin{align}
    t \ge \big( 2\theta e k \log n \big) (1+\eta) \label{eq:dd_t2}
\end{align}
for arbitrarily small $\eta > 0$.

\subsection{Combining and Simplifying} \label{sec:dd_combining}

To complete the proof of Theorem \ref{thm:dd}, we only need to show that given the requirements \eqref{eq:dd_t1_1} and \eqref{eq:dd_t2}, the additional requirement \eqref{eq:dd_t1_2} is redundant (recall also that $k = \kbar(1+o(1))$ for any typical graph $G$).  We handle the three cases separately:
\begin{itemize}[leftmargin=5ex,itemsep=0ex,topsep=0.25ex]
    \item For the first case $\frac{3}{4} \le \theta \le 1$, observe that the coefficient $2 - \theta \le 1.25$ in \eqref{eq:dd_t1_2} is strictly less than the coefficient $2\theta \ge 1.5$ in \eqref{eq:dd_t2}.
    \item For the second case $\frac{1}{2} < \theta < \frac{3}{4}$, observe that the coefficient $3\theta - 1 < 2\theta - 0.25$ in \eqref{eq:dd_t1_2} is strictly less than the coefficient $2\theta$ in \eqref{eq:dd_t2}. 
    \item For the third case $0 < \theta \le \frac{1}{2}$, observe that the coefficient $\theta \le \frac{1}{2}$ in \eqref{eq:dd_t1_2} is strictly less than the coefficient $2(1-\theta) \ge 1$ in \eqref{eq:dd_t1_1}.
\end{itemize}

\subsection{Derivation of the Conditional Distribution \eqref{eq:Tij}} \label{sec:dd_cond_distr}

The derivation of \eqref{eq:Tij} is based on multinomial conditioning, and bears similarity to an analogous conditional distribution for standard group testing \cite[Sec.~A.3]{Ald14a}.  To derive the conditional distribution given $t_+$ and $\peset$, we first need to consider certain unconditional distributions (though still with implicit conditioning on a given typical graph $G$).  We define the following random variables:
\begin{itemize}
    \item $T_-$ is the number of negative tests, $\Ttil_{i,j}$ is the number of tests covering a given true edge $(i,j) \in E$ but no other true edges, and $\Ttilmul$ is the number of tests covering two or more true edges.
    \item $T_{i,j}$ is the number of tests covering a given true edge $(i,j) \in E$ and no other pairs from $\PE$.
\end{itemize}
Since the tests are independent, $(T_-, \{\Ttil_{i,j}\}_{(i,j) \in E},\Ttilmul)$ has a multinomial distribution with $t$ trials; the corresponding probability parameters are denoted by $(q_-, \{\qtil_{i,j}\}_{(i,j) \in E},\qtilmul)$.

We now consider conditioning on $T_- = t_-$ and $\PE = \peset$.  Under such conditioning, we can characterize the joint distribution of $(\{T_{i,j}\}_{(i,j) \in E}, \{\Ttil_{i,j} - T_{i,j}\}_{(i,j) \in E},\Ttilmul)$ via the following lemma from \cite{Ald14a}, stated in generic notation.

\begin{lemma} \label{lem:multi_cond}
    {\em \cite[Lemma C.1]{Ald14a}}
    Fix the integers $\ell$ and $m$, and let $(W_0,\{W_i\}_{i=1}^{\ell},W_{\ell+1})$ have a multinomial distribution with $m$ trials and probabilities $(r_0,\{r_i\}_{i=1}^{\ell},r')$.  Associate an observation $(W_0,\{W_i\}_{i=1}^{\ell},W') = (w_0,\{w_i\}_{i=1}^{\ell},w')$ with an unordered list of $m$ class labels (class $0$, class $i=1,\dotsc,\ell$, or class $\ell+1$), and suppose that each label in class $i=1\dotsc,m$ is independently changed to some class $i'$ with probability $\gamma_i \in [0,1]$, and to some class $i''$ with probability $1-\gamma_i$ (where $\gamma_i$ may depend on $w_0$).  Then, conditioned on $W_0 = w_0$, the corresponding random variables $(\{W'_i\}_{i=1}^{\ell},\{W''_i\}_{i=1}^{\ell},W_{\ell+1})$ counting the transformed class labels have a multinomial distribution with $m-w_0$ trials and the following probability parameters:
    \begin{equation}
        \bigg( \Big\{\frac{r_i \gamma_i}{1-r_0}\Big\}_{i=1}^{\ell}, \Big\{\frac{r_i (1-\gamma_i)}{1-r_0}\Big\}_{i=1}^{\ell}, \frac{r'}{1-r_0} \bigg).
    \end{equation}
\end{lemma}

% Both of these are determined entirely by the negative tests, meaning that the positive tests remain conditionally independent, but with a different distribution.  The multinomial probabilities are all divided by $q_-$ due to the conditioning on $T_- = t_-$ \cite[Lemma C.1 part 3]{Ald14a}, and each probability-$\qtil_{i,j}$ event (corresponding to adding one to $\Ttil_{i,j}$) is ``kept'' (contributing to $T_{i,j}$) independently with some conditional probability $\gamma_{i,j}$, and ``not kept'' (contributing to $\Ttil_{i,j} - T_{i,j}$) with probability $1-\gamma_{i,j}$.  Specifically, $\gamma_{i,j}$ is the conditional probability that some non-edge in $\peset$ is covered by the test given that $(i,j)$ is the unique true edge covered.  This ``splitting'' of the $\Ttil_{i,j}$ random variables again leads to to a multinomial distribution \cite[Lemma C.1 part 4]{Ald14a}: G

To apply this result, we associate $(T_-, \{\Ttil_{i,j}\}_{(i,j) \in E},\Ttilmul)$ with $(W_0,\{W_i\}_{i=1}^{\ell},W_{\ell+1})$, and associate  $(\{T_{i,j}\}_{(i,j) \in E}, \{\Ttil_{i,j} - T_{i,j}\}_{(i,j) \in E},\Ttilmul)$ with $(\{W'_i\}_{i=1}^{\ell},\{W''_i\}_{i=1}^{\ell},W_{\ell+1})$.  Conditioning on $T_- = t_-$ amounts to conditioning on $W_0$, and conditioning on $\PE = \peset$ only amounts to changing the value of $\gamma_i$, since $\PE$ is determined entirely by the negative tests.  Notice that any test contributing (i.e., adding one) to $\Ttil_{i,j}$ further contributes to $T_{i,j}$ independently with probability  $\gamma_{i,j}$, defined to be the conditional probability that some non-edge in $\peset$ is covered by the test given that $(i,j)$ is the unique true edge covered.  Hence, given $T_- = t_-$ and $\PE = \peset$, Lemma \ref{lem:multi_cond} implies that the random variables $(\{T_{i,j}\}_{(i,j) \in E}, \{\Ttil_{i,j} - T_{i,j}\}_{(i,j) \in E},\Ttilmul)$ have a multinomial distribution with $t_+ = t - t_-$ trials and the following probability parameters:
\begin{itemize}[leftmargin=5ex,itemsep=0ex,topsep=0.25ex]
    \item For $(i,j) \in E$, the parameter for $T_{i,j}$ is $\frac{ \qtil_{i,j} \gamma_{i,j} }{ 1 - q_- }$;
    \item For $(i,j) \in E$, the parameter for $\Ttil_{i,j} - T_{i,j}$ is $\frac{ \qtil_{i,j} (1-\gamma_{i,j}) }{ 1 - q_- }$;
    \item The parameter for $\Ttilmul$ is $\frac{ \qtilmul }{ 1 - q_- }$.
\end{itemize}
We conclude by showing that \eqref{eq:Tij} follows from the first of these dot points, with the marginal distribution of a multinomial distribution being binomial.  The denominator $1-q_-$ is trivially equal to $q_+$, and the numerator $\qtil_{i,j} \gamma_{i,j}$ equals the product of two terms. To understand these terms, let $\Btil_{i,j}$ be the event that a given test covers $(i,j)$ but no other true edge, and let $B_{i,j}$ be the event that it covers $(i,j)$ but no other pair from $\PE$.  Then, the previous definitions can be written as
\begin{gather}
    \qtil_{i,j} = \PP[ \Btil_{i,j} ], \quad \gamma_{i,j} = \PP[ B_{i,j} \,|\, \peset, \Btil_{i,j} ]. 
\end{gather}
In addition, we have $\PP[ \Btil_{i,j} ] = \PP[ \Btil_{i,j} \,|\, \peset ]$, since $\Btil_{i,j}$ is independent of $\PE$ (note that $\PE$ is determined entirely by the negative tests, and $\Btil_{i,j}$ only concerns true edges).   As a result, we have
\begin{align}
    \qtil_{i,j}\gamma_{i,j}    
        &= \PP[ \Btil_{i,j} \,|\, \peset ]\PP[ B_{i,j} \,|\, \peset, \Btil_{i,j} ] \\
        &= \PP[ \Btil_{i,j} \cap B_{i,j} \,|\, \peset ] \\
        &= \PP[ B_{i,j} \,|\, \peset ], \label{eq:q_gamma}
\end{align}
where \eqref{eq:q_gamma} follows since $B_{i,j}$ implies $\Btil_{i,j}$, because all true edges are in $\PE$ with probability one (i.e., the first step of DD has no false negatives).  Finally, \eqref{eq:q_gamma} coincides precisely with the definition of $q_{i,j}$ stated following \eqref{eq:Tij}, and this completes the derivation of \eqref{eq:Tij}.

\section{Proof of Theorem \ref{thm:sss} (SSS Lower Bound)} \label{sec:pf_sss}

Since the random graph is in the typical set \eqref{eq:typical} with probability approaching one, it suffices to establish that the number of tests \eqref{eq:sss} yields error probability tending to one conditioned on an arbitrary typical graph $G \in \Tc_n(\epsilon_n)$, with $\epsilon_n = o(1)$ due to Lemma \ref{lem:typical}.  We implicitly condition on such a graph $G$ throughout the analysis.

Let $M_{ij}$ be the event that edge $(i,j)$ is {\em masked}, i.e., whenever its nodes both appear in a test, the nodes of some different edge are also included in the test.  In this case, there exists a satisfying set (of edges) of cardinality $k-1$, so the algorithm will fail to output the true edge set.  Hence,
\begin{align}
    \pe &\ge \PP\bigg[ \bigcup_{(i,j) \in E} M_{ij} \bigg] \\
        &\ge \sum_{(i,j) \in E} \frac{ \PP[M_{ij}]^2 }{ \sum_{(i',j') \in E} \PP[M_{ij} \cap M_{i'j'}]  }, \label{eq:deCaen2}
\end{align}
where \eqref{eq:deCaen2} is an application of de Caen's bound \cite{Dec97}.

We proceed by bounding the individual and pairwise masking probabilities.  For a given edge $(i,j)$ to be masked, for each of the $t$ tests we need either $i$ or $j$ to be excluded, or for the nodes of some other edge to be included.  Letting $A^{(ij)}_1$ be the event that some other node connected to $i$ or $j$ is included, and $A^{(ij)}_2$ the event that two connected nodes distinct from $i$ and $j$ are included, the associated masking event for a single test has probability
\begin{eqnarray}
    p_1^{(ij)} = (1-p^2) + p^2 \PP[ A^{(ij)}_1 \cup A^{(ij)}_2 \,|\, \{i,j\} \subseteq \Lc ]. \label{eq:pij}
\end{eqnarray}
We lower bound $p_1^{(ij)}$ by ignoring the event $A^{(ij)}_1$:
\begin{align}
    p_1^{(ij)} 
        &\ge (1-p^2) + p^2 \PP[ A^{(ij)}_2 \,|\, \{i,j\} \subseteq \Lc ] \\
        &= 1 - p^2 + p^2 \PP[ A^{(ij)}_2 ], \label{eq:pij_lower}
\end{align}
since $A^{(ij)}_2$ is independent of whether $\{i,j\} \subseteq \Lc$.  Now observe that the unconditional probability of a positive test satisfies
\begin{align}
    \PP[Y=1] 
        &= \PP\big[ \{\{i,j\} \subseteq \Lc\} \cup A^{(ij)}_1 \cup A^{(ij)}_2 \big] \label{eq:PY1} \\
        &\le \PP\big[ \{i,j\} \subseteq \Lc \big] + \PP\big[ A^{(ij)}_1 \big] + \PP\big[ A^{(ij)}_2 \big] \\
        &\le p^2 + 2dp^2 + \PP\big[ A^{(ij)}_2 \big], 
\end{align}
and hence
\begin{equation}
    \PP\big[ A^{(ij)}_2 \big] \ge P_Y(1) - \xi, \label{eq:xi}
\end{equation}
where $P_Y(1)$ is a shorthand for $\PP[Y=1]$, and $\xi = (1+2d)p^2$.  Substitution into \eqref{eq:pij_lower} gives 
\begin{equation}
    p_1^{(ij)} \ge 1 - p^2 (1 - P_Y(1) + \xi).
\end{equation}

Next, we upper bound $p_1^{(ij)}$.  Applying the union bound in \eqref{eq:pij}, we obtain 
\begin{align}
    p_1^{(ij)} 
        &\le (1-p^2) + p^2 \big( \PP[ A^{(ij)}_1 \,|\, \{i,j\} \subseteq \Lc ] + \PP [A^{(ij)}_2 \,|\, \{i,j\} \subseteq \Lc ] \big) \\
        &\le (1-p^2) + p^2 \big( 2dp + \PP [A^{(ij)}_2 ] \big) \label{eq:pij_upper2} \\
        &\le 1 - p^2 \big( 1 - P_Y(1) - \xi'\big) \label{eq:pij_upper3} 
\end{align}
where \eqref{eq:pij_upper2} uses $\PP[ A^{(ij)}_1 \,|\, \{i,j\} \subseteq \Lc ] \le 2dp$ and the fact that $A^{(ij)}_2$ is independent of whether $\{i,j\} \subseteq \Lc$, and \eqref{eq:pij_upper3} uses $\PP [A^{(ij)}_2 ] \le P_Y(1)$ (see \eqref{eq:PY1}) and defines $\xi' = 2dp$.  

Now, for the masking event $M_{ij}$ to occur, the probability-$p_1^{(ij)}$ masking event needs to occur for all tests, yielding $\PP[M_{ij}] = (p_1^{(ij)})^t$.  Moreover, for both $M_{ij}$ and $M_{i'j'}$ to occur, the case $(i,j) = (i',j')$ is handled trivially, whereas for $(i,j) \ne (i',j')$ the associated events for $(i,j)$ and $(i',j')$ need to occur simultaneously for each test.  Since the complementary event (i.e., the edge is the only one covered by the nodes included in the test) can only occur for one of $(i,j)$ or $(i',j')$, the associated probability $p_1^{(ij \,\cap\, i'j')}$ of both masking events occurring for a single test satisfies
\begin{equation}
    1 - p_1^{(ij \,\cap\, i'j')} = (1-p_1^{(ij)}) +  (1-p_1^{(i'j')}),
\end{equation}
i.e., $\PP[A \cup  B] = \PP[A] + \PP[B]$ for disjoint events $A$ and $B$.  Hence, from \eqref{eq:pij_upper3},
\begin{equation}
    p_1^{(ij \,\cap\, i'j')} \le 1 - 2 p^2 \big( 1 - P_Y(1) - \xi'\big).
\end{equation}
Taking the intersection over the $t$ tests gives $\PP[M_{ij} \cap M_{i'j'}] = (p_1^{(ij \,\cap\, i'j')})^t$ for all $(i,j) \ne (i',j')$, and substituting the preceding findings into \eqref{eq:deCaen2} gives 
\begin{align}
    \PP[{\rm error}] 
        &\ge \sum_{(i,j) \in E} \frac{ \big( p_1^{ij} \big)^{2t} }{ \big( p_1^{ij} \big)^{t} + \sum_{(i',j') \ne (i,j)} \big(p_1^{(ij \,\cap\, i'j')} \big)^t  } \\
        &\ge \sum_{(i,j) \in E} \frac{ \big( 1 - p^2 (1 - P_Y(1) + \xi) \big)^{2t} }{ \big( 1 - p^2 ( 1 - P_Y(1) - \xi') \big)^{t} + \sum_{(i',j') \ne (i,j)} \big(1 - 2 p^2 ( 1 - P_Y(1) - \xi') \big)^t  } \\
        &\ge \frac{ k \big( 1 - p^2 (1 - P_Y(1) + \xi) \big)^{2t} }{ \big( 1 - p^2 ( 1 - P_Y(1) - \xi') \big)^{t} + k \big(1 - 2 p^2 ( 1 - P_Y(1) - \xi') \big)^t  }, \label{eq:post_dec}
\end{align}
since $|E| = k$ (in the denominator, we upper bound $k-1 \le k$).

We upper bound the terms in the denominator in \eqref{eq:post_dec} using $1-\alpha \le e^{-\alpha}$, and characterize the numerator using $1-\alpha = e^{-\alpha + O(\alpha^2)}$ as $\alpha \to 0$ (recall that $p^2 = \Theta\big(\frac{1}{k}\big) = o(1)$ and $P_Y(1) = \Theta(1)$):
\begin{align}
    \PP[{\rm error}] 
        &\ge \frac{ k e^{- 2t \big( p^2 (1 - P_Y(1) + \xi) + O(p^4) \big) } }{ e^{- t p^2 ( 1 - P_Y(1) - \xi')} + k e^{- 2 t p^2 ( 1 - P_Y(1) - \xi')}  } \\
        &\ge \frac{ k e^{- t  p^2 (1 - P_Y(1)) } }{ 1 + k e^{- t p^2 ( 1 - P_Y(1))}  } \cdot \frac{e^{-2t(p^2 \xi + O(p^4))}}{ e^{2tp^2 \xi'} }. \label{eq:Perr_lb_2}
\end{align}
Since the converse bound we are proving is of the form $t = \Omega(k \log n)$, we can assume without loss of generality that $t = \Theta(k \log n)$, as additional tests can only help the SSS algorithm.\footnote{If an incorrect set is satisfying with respect to a certain number of tests, it remains satisfying after removing any subset of those tests.  Hence, removing tests cannot decrease the error probability.}  In addition, we can assume without loss of generality that $p^2 = \Theta\big(\frac{1}{k}\big)$, since if $p^2$ behaves as $o\big(\frac{1}{k}\big)$ or $\omega\big(\frac{1}{k}\big)$ then the probability of a positive test tends to $0$ or $1$ as $n \to \infty$, and it follows from a standard entropy-based argument that  $\omega(k \log n)$ tests are needed \cite[Lemma 1]{Ald15}.   We claim that these conditions imply that
\begin{equation}
    \frac{e^{-2t(p^2 \xi + O(p^4))}}{ e^{2tp^2 \xi'} } \to 1. \label{eq:lim1}
\end{equation}
This is seen by noting that $tp^2 = \Theta( \log n )$ by the above-mentioned behavior of $t$ and $p^2$, whereas the terms $\xi = (1+2d)p^2$, $\xi' = 2dp$, and $O(t p^4)$ all behave as $O(n^{-c})$ for sufficiently small $c$.  This behavior is easy to see for the $O(t p^4)$ term by the above-mentioned behavior of $t$ and $p^2$, and is seen to also hold for $\xi$ and $\xi'$ by noting that $dp = \Theta\big(\frac{d}{\sqrt k}\big)$, along with $d = O( \max\{ \log n, nq \} )$ (see \eqref{eq:d_bound}), $\sqrt{k} = \Theta( n\sqrt{q} )$, and the behavior of $q$ in \eqref{eq:q_bounds}.

Substituting \eqref{eq:lim1} into \eqref{eq:Perr_lb_2}, we have
\begin{equation}
    \PP[{\rm error}] \ge \frac{1}{1 + ke^{tp^2(1-P_Y(1))}} (1+o(1)),
\end{equation}
and substituting $p^2 = \frac{\nu}{k}$ and $1- P_Y(1) = e^{-\nu} (1+o(1))$, we deduce that $\PP[{\rm error}] \to 1$ whenever
\begin{equation}
    t \le \frac{k \log k}{\nu e^{-\nu}} (1 - \eta)
\end{equation}
for arbitrarily small $\eta > 0$.  Since the function $\nu e^{-\nu}$ is maximized at $\nu = 1$, we deduce that $\PP[{\rm error}] \to 1$ whenever
\begin{equation}
    t \le \big( k e \log k \big) (1 - \eta).
\end{equation}
The proof is completed by recalling that for any typical graph, $k = \kbar(1+o(1))$ and $\log k = (2\theta \log n)(1+o(1))$ (since $\kbar = \Theta(n^{2\theta})$).

\section{Missing Details in the Proof of Theorem \ref{thm:grotesque} (Sublinear-Time Decoding)} \label{sec:bundles}

\subsection{Details of Step 1 -- Bundles of Tests}

% Here we provide the omitted details from Section \ref{sec:bundles_main}.  
Recall that we form a number $B$ of ``bundles'' of tests, where each node is placed in each bundle with probability $r \in (0,1)$.
For a given bundle, consider the probability $p_{\rm one}$ of a given edge $(i,j)$ being the only edge among its nodes.  Letting $A_1$ be the event that some other node connected to $i$ or $j$ is in the bundle, and letting $A_2$ be the event that two different edge-connected nodes are in the bundle, we have
\begin{align}
    p_{\rm one} 
        &= r^2 \cdot \PP[A_1^c \cap A_2^c] \\
        &\ge r^2 \cdot\big( 1 - \PP[A_1] - \PP[A_2] \big) \\
        &\ge r^2 \cdot\big( 1 - 2dr - kr^2 \big), \label{eq:pone3}
\end{align}
where \eqref{eq:pone3} uses the fact that there are at most $2d$ nodes connected to $i$ or $j$, and at most $k$ other edges separate from $i$ and $j$.  Setting $r = \frac{1}{\sqrt{2\kbar}}$ gives
\begin{align}
    p_{\rm one}
        &\ge \frac{1}{2\kbar} \Big( 1 - \frac{2d}{\sqrt{2\kbar}} - \frac{1}{2}\cdot\frac{k}{\kbar} \Big) \\
        &= \frac{1}{4\kbar} (1+o(1)),
\end{align}
since $d \ll \sqrt{k}$ (see \eqref{eq:d_vs_k}) and $k = \kbar(1+o(1))$.  Hence, the probability of $(i,j)$ being the unique edge in {\em some} bundle satisfies
\begin{align*}
    p_{\rm any} 
        &= 1 - \bigg( 1 - \frac{1}{4\kbar} (1+o(1)) \bigg)^B \\
        &\ge 1 - e^{- \frac{B}{4\kbar} (1+o(1))}, 
\end{align*}
and by a union bound over the $k = \kbar(1+o(1))$ edges, we find that $B = \big(4\kbar \log \kbar\big) (1+o(1))$ bundles suffice to ensure that every edge is the unique one in at least one bundle.

\subsection{Details of Step 4 -- Total Number of Tests and Decoding Time}

The number of tests used is asymptotically dominated by that of the location tests, and recalling that $B = \big(4\kbar \log \kbar\big) (1+o(1))$, we find that $t = 4\kbar (\log \kbar)(\log_2 n)^2 (1+o(1))$.  We briefly mention that this can be significantly reduced when adaptivity is allowed, similarly to standard group testing \cite{Cai13}, but our focus in this paper is on the non-adaptive setting.

For the decoding time, we notice that each multiplicity test takes $O(\log B)$ time (i.e., the same as the number of tests used), whereas for the location test we can actually make the decoding time less than the number of tests due to the fact that we don't end up making use of most test outcomes.\footnote{We still need to perform such tests, because {\em a priori} we don't know which $O(L)$ size subset of the $O(L^2)$ relevant tests performed will end up being useful. }  Specifically, each iteration from $\ell = 1,\dotsc,L$ observes at most $3$ test outcomes and runs in $O(1)$ time, so the overall time per location test is $O(L)$.  Since the decoder makes use of $B$ multiplicity tests and $k$ location tests (assuming no errors occur), the total runtime is $O(B \log B + O(k L))$, which simplifies to $O( \kbar \log^2 \kbar + \kbar \log n )$.

\section{Additional Numerical Experiments} \label{app:more_exp}

In order to demonstrate that the empirical performance of our algorithms is in agreement with our theory, we plot the success probability as a function of the number of tests for various $(\kbar,n)$ pairs, and then re-plot them with the number of tests normalized by $\overline{k}\log\frac{1}{q}$ (e.g., see Theorem \ref{thm:conv}; similar normalization is also used in Figure \ref{fig:rates}).  The results, averaged over $1000$ trials, are shown in Figure \ref{fig:more_exp}.

\begin{figure}
    \begin{centering}
        {% This file was created by matlab2tikz.
\definecolor{dark green}{rgb}{0,0.7,0}%
\begin{tikzpicture}

\begin{axis}[%
width=5.2cm,
height=4cm,
scale only axis,
xmin=25,
xmax=800,
xtick={  100,  200,  300,  400,  500,  600,  700, 800},
xlabel style={font=\color{white!15!black}},
xlabel={Number of tests},
ymin=0,
ymax=1,
ytick={  0, 0.2, 0.4, 0.6, 0.8,   1},
ylabel style={font=\color{white!15!black}},
ylabel={Success probability},
axis background/.style={fill=white},
%axis x line*=bottom,
%axis y line*=left,
xmajorgrids,
ymajorgrids,
legend style={at={(0.99,0.4)}, anchor=south east, legend cell align=left, align=left, draw=white!15!black, font= \scriptsize,inner xsep=1pt, inner ysep=1pt}
]
\addplot [color=red, line width=1pt]
  table[row sep=crcr]{%
25	0.008\\
50	0.035\\
75	0.164\\
100	0.373\\
125	0.549\\
150	0.688\\
175	0.774\\
200	0.858\\
225	0.894\\
250	0.92\\
275	0.943\\
300	0.967\\
325	0.976\\
350	0.986\\
375	0.982\\
400	0.986\\
425	0.996\\
450	0.997\\
475	0.995\\
500	0.996\\
525	0.998\\
550	0.998\\
575	1\\
600	0.999\\
625	1\\
650	0.999\\
675	0.999\\
700	1\\
725	0.999\\
750	1\\
775	1\\
800	0.999\\
};
\addlegendentry{LP}

\addplot [color=blue, line width=1pt]
  table[row sep=crcr]{%
25	0.001\\
50	0.002\\
75	0.047\\
100	0.16\\
125	0.31\\
150	0.448\\
175	0.589\\
200	0.694\\
225	0.755\\
250	0.838\\
275	0.877\\
300	0.898\\
325	0.924\\
350	0.957\\
375	0.952\\
400	0.97\\
425	0.972\\
450	0.987\\
475	0.986\\
500	0.991\\
525	0.989\\
550	0.994\\
575	0.991\\
600	0.995\\
625	1\\
650	0.999\\
675	0.996\\
700	1\\
725	0.999\\
750	0.998\\
775	0.998\\
800	0.999\\
};
\addlegendentry{DD}

\addplot [color=dark green, line width=1pt]
  table[row sep=crcr]{%
25	0\\
50	0\\
75	0.003\\
100	0.032\\
125	0.09\\
150	0.17\\
175	0.276\\
200	0.355\\
225	0.45\\
250	0.57\\
275	0.63\\
300	0.691\\
325	0.737\\
350	0.791\\
375	0.816\\
400	0.849\\
425	0.868\\
450	0.909\\
475	0.916\\
500	0.932\\
525	0.937\\
550	0.957\\
575	0.96\\
600	0.969\\
625	0.976\\
650	0.978\\
675	0.977\\
700	0.98\\
725	0.987\\
750	0.984\\
775	0.986\\
800	0.983\\
};
\addlegendentry{COMP}

\addplot [color=red, dash pattern= on 5pt off 4pt, line width=1pt]
  table[row sep=crcr]{%
25	0\\
50	0.005\\
75	0.052\\
100	0.174\\
125	0.299\\
150	0.496\\
175	0.64\\
200	0.72\\
225	0.809\\
250	0.872\\
275	0.926\\
300	0.954\\
325	0.949\\
350	0.969\\
375	0.978\\
400	0.981\\
425	0.987\\
450	0.985\\
475	0.995\\
500	0.996\\
525	0.998\\
550	0.999\\
575	0.998\\
600	0.998\\
625	0.998\\
650	1\\
675	1\\
700	1\\
725	1\\
750	1\\
775	1\\
800	1\\
};
\addplot [color=red, dash pattern=on 2pt off 3pt on 5pt off 3pt, line width=1pt]
  table[row sep=crcr]{%
25	0\\
50	0.002\\
75	0.011\\
100	0.056\\
125	0.147\\
150	0.282\\
175	0.429\\
200	0.572\\
225	0.688\\
250	0.777\\
275	0.83\\
300	0.887\\
325	0.924\\
350	0.937\\
375	0.96\\
400	0.956\\
425	0.97\\
450	0.983\\
475	0.988\\
500	0.988\\
525	0.995\\
550	0.997\\
575	0.994\\
600	0.998\\
625	1\\
650	0.996\\
675	0.999\\
700	0.999\\
725	0.999\\
750	0.999\\
775	0.998\\
800	1\\
};
\addplot [color=red, dash pattern=on 2pt off 3pt, line width=1pt]
  table[row sep=crcr]{%
25	0\\
50	0\\
75	0.001\\
100	0.007\\
125	0.046\\
150	0.143\\
175	0.258\\
200	0.4\\
225	0.532\\
250	0.584\\
275	0.722\\
300	0.78\\
325	0.824\\
350	0.879\\
375	0.924\\
400	0.934\\
425	0.961\\
450	0.972\\
475	0.977\\
500	0.979\\
525	0.988\\
550	0.993\\
575	0.994\\
600	0.993\\
625	0.992\\
650	0.996\\
675	0.996\\
700	0.995\\
725	0.998\\
750	1\\
775	0.999\\
800	0.999\\
};
\addplot [color=blue, dash pattern= on 5pt off 4pt, line width=1pt]
  table[row sep=crcr]{%
25	0\\
50	0\\
75	0.007\\
100	0.034\\
125	0.084\\
150	0.232\\
175	0.379\\
200	0.496\\
225	0.593\\
250	0.677\\
275	0.788\\
300	0.821\\
325	0.869\\
350	0.915\\
375	0.936\\
400	0.951\\
425	0.963\\
450	0.96\\
475	0.975\\
500	0.984\\
525	0.989\\
550	0.991\\
575	0.989\\
600	0.989\\
625	0.995\\
650	0.996\\
675	0.997\\
700	0.998\\
725	0.999\\
750	0.998\\
775	0.999\\
800	0.998\\
};
\addplot [color=blue, dash pattern=on 2pt off 3pt on 5pt off 3pt, line width=1pt]
  table[row sep=crcr]{%
25	0\\
50	0\\
75	0\\
100	0.003\\
125	0.023\\
150	0.062\\
175	0.171\\
200	0.305\\
225	0.426\\
250	0.517\\
275	0.614\\
300	0.705\\
325	0.798\\
350	0.828\\
375	0.858\\
400	0.893\\
425	0.922\\
450	0.953\\
475	0.958\\
500	0.95\\
525	0.984\\
550	0.982\\
575	0.981\\
600	0.992\\
625	0.987\\
650	0.995\\
675	0.994\\
700	0.993\\
725	0.998\\
750	0.999\\
775	0.997\\
800	0.997\\
};
\addplot [color=blue, dash pattern=on 2pt off 3pt, line width=1pt]
  table[row sep=crcr]{%
25	0\\
50	0\\
75	0\\
100	0\\
125	0\\
150	0.018\\
175	0.06\\
200	0.115\\
225	0.222\\
250	0.288\\
275	0.43\\
300	0.523\\
325	0.603\\
350	0.707\\
375	0.757\\
400	0.821\\
425	0.861\\
450	0.892\\
475	0.905\\
500	0.925\\
525	0.951\\
550	0.969\\
575	0.976\\
600	0.98\\
625	0.984\\
650	0.98\\
675	0.995\\
700	0.994\\
725	0.994\\
750	0.997\\
775	0.998\\
800	0.998\\
};
\addplot [color=dark green, dash pattern= on 5pt off 4pt, line width=1pt]
  table[row sep=crcr]{%
25	0\\
50	0\\
75	0\\
100	0.003\\
125	0.007\\
150	0.054\\
175	0.116\\
200	0.165\\
225	0.261\\
250	0.298\\
275	0.446\\
300	0.499\\
325	0.574\\
350	0.664\\
375	0.69\\
400	0.729\\
425	0.783\\
450	0.791\\
475	0.842\\
500	0.881\\
525	0.898\\
550	0.923\\
575	0.916\\
600	0.916\\
625	0.949\\
650	0.959\\
675	0.95\\
700	0.961\\
725	0.971\\
750	0.974\\
775	0.974\\
800	0.978\\
};
\addplot [color=dark green, dash pattern=on 2pt off 3pt on 5pt off 3pt, line width=1pt]
  table[row sep=crcr]{%
25	0\\
50	0\\
75	0\\
100	0\\
125	0\\
150	0.006\\
175	0.029\\
200	0.044\\
225	0.094\\
250	0.124\\
275	0.235\\
300	0.31\\
325	0.392\\
350	0.446\\
375	0.487\\
400	0.561\\
425	0.644\\
450	0.711\\
475	0.755\\
500	0.761\\
525	0.807\\
550	0.849\\
575	0.864\\
600	0.907\\
625	0.895\\
650	0.92\\
675	0.934\\
700	0.933\\
725	0.936\\
750	0.955\\
775	0.964\\
800	0.972\\
};
\addplot [color=dark green, dash pattern=on 2pt off 3pt, line width=1pt]
  table[row sep=crcr]{%
25	0\\
50	0\\
75	0\\
100	0\\
125	0\\
150	0\\
175	0.002\\
200	0.005\\
225	0.022\\
250	0.042\\
275	0.074\\
300	0.122\\
325	0.181\\
350	0.293\\
375	0.333\\
400	0.403\\
425	0.475\\
450	0.512\\
475	0.581\\
500	0.627\\
525	0.692\\
550	0.753\\
575	0.772\\
600	0.797\\
625	0.828\\
650	0.874\\
675	0.887\\
700	0.901\\
725	0.911\\
750	0.927\\
775	0.935\\
800	0.938\\
};
\end{axis}

\begin{axis}[%
width=5.2cm,
height=4cm,
scale only axis,
xmin=25,
xmax=800,
xtick={\empty},
ymin=0,
ymax=1,
ytick={  \empty},
axis line style={draw=none},
legend style={at={(0.99,0.015)}, anchor=south east, legend cell align=left, align=left, draw=white!15!black, fill=white, font= \scriptsize,inner xsep=1pt, inner ysep=1pt}
]
\addplot [color=black, line width=1pt]
  table[row sep=crcr]{%
-1	0\\
0	0\\
};
\addlegendentry{$n = 80$}

\addplot [color=black, dash pattern= on 5pt off 4pt, line width=1pt]
  table[row sep=crcr]{%
-1	0\\
0	0\\
};
\addlegendentry{$n = 100$}

\addplot [color=black, dash pattern=on 2pt off 3pt on 5pt off 3pt, line width=1pt]
  table[row sep=crcr]{%
-1	0\\
0	0\\
};
\addlegendentry{$n = 120$}

\addplot [color=black, dash pattern=on 2pt off 3pt, line width=1pt]
  table[row sep=crcr]{%
-1	0\\
0	0\\
};
\addlegendentry{$n = 140$}

\end{axis}
\end{tikzpicture}%
~~~
        \input{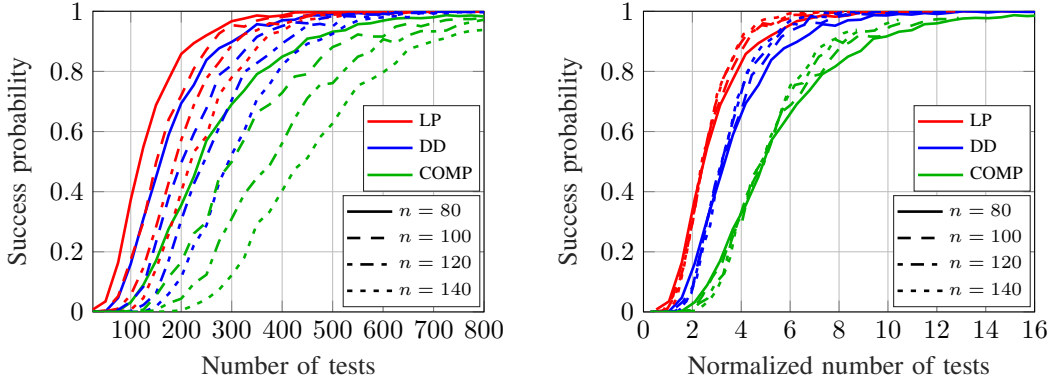}}
        \par
    \end{centering}
    \caption{(Left) Number of tests for COMP, DD, and LP under four different $(n,\overline{k})$ pairs: $n \in \{80,100,120,140\}$ and $\overline{k} = \frac{n}{10}$.  (Right) Normalized number of tests after division by $\overline{k}\log\frac{1}{q}$, where $q$ is the probability of each edge in the graph. \label{fig:more_exp}}
\end{figure} 

As predicted by our theory, the resulting curves for each algorithm are in general agreement after performing the normalization, with slight deviations due to noise and non-asymptotic considerations.  Moreover, according to the sparse regime of Figure \ref{fig:rates}, our theory suggests (for sufficiently sparse settings) an asymptotic threshold of roughly $1$ for the optimal algorithm (which LP approximates), roughly $2$ for DD, and slightly over $2$ for COMP.  The above figure is consistent with these numbers, though they are slightly increased because of the penalty incurred for finite $n$ (as opposed to $n \to \infty$), and possibly also the choice $\overline{k} = \frac{n}{10}$ (as opposed to  $\overline{k} \sim n^{2\theta}$).

\section{Results for General Edge and Degree Bounded Graphs} \label{sec:gen_d_k}

Since the assumption of independent edges is not always appropriate for modeling real-world applications, there is substantial motivation to develop performance bounds that hold with high probability for any given graph in a {\em deterministic graph class}.  

An impossibility result of \cite{Aba18} shows that if one only fixes the number of edges to $k$, then achieving $t = O(k \log n)$ scaling is not possible in the worst case.  A natural question is then whether fixing the number of edges $k$ {\em and} maximum degree $d$ results in $t = O(k \log n)$ scaling under suitable assumptions on $d$.  In this section, we argue that the answer is affirmative as long as $d = o(\sqrt{k})$.

Indeed, an inspection of our analysis reveals that the condition $d \le d_{\max}$ in the typical set \eqref{eq:typical} was not used directly, but rather, was only used to establish \eqref{eq:d_vs_k}.  On the other hand, the condition $\PP_G[Y=1] = (1-e^{-\nu})(1+o(1))$ played a significant role in our analysis, and it is unclear whether it can be deduced from the condition $d = o(\sqrt{k})$ alone.

However, while exactly characterizing $\PP_G[Y=1]$ for an arbitrary edge-bounded and degree-bounded graph $G$ may be difficult, we can easily find upper and lower bounds.  First, by the union bound, we have
\begin{equation}
    \PP_G[Y=1] \le k p^2 = \nu,
\end{equation}
under the choice $p = \sqrt{\frac{\nu}{k}}$.  As for the lower bound, applying de Caen's bound \cite{Dec97} and letting $A_1,\dotsc,A_k$ be the events of the $k$ edges having both their nodes included in the test, we have
\begin{align}
    \PP_G[Y=1] 
        &= \PP_G\bigg[ \bigcup_{i=1,\dotsc,k} \big\{ A_i \big\} \bigg] \\
        &\ge \sum_{i=1,\dotsc,k} \frac{ \PP_G[A_i]^2 }{ \PP_G[A_i] + \sum_{j \ne i} \PP_G[A_i \cap A_j] }.
\end{align}
Note that $\PP_G[A_i] = p^2$ for all $i$, since the two nodes of the edge need to be included.

Among the terms $\sum_{j \ne i} \PP_G[A_i \cap A_j]$, there are at most $2d$ terms for which the two associated edges share a node ($d$ per node times two nodes), and for those terms we have $\PP_G[A_i \cap A_j] = p^3$.  All other terms (of which there are at most $k$) have $\PP_G[A_i \cap A_j] = p^4$, and hence
\begin{align}
    \PP_G[Y=1] 
        &\ge \frac{ k p^4 }{ p^2 + 2dp^3 + kp^4 } \\
        &= \frac{\nu}{1+\nu} (1+o(1)),
\end{align}
where we have used $p = \sqrt{\frac{\nu}{k}}$ and $dp \ll 1$.

With the above upper and lower bounds on $\PP_G[Y=1]$ in place, upon fixing $\nu \in (0,1)$ (e.g., $\nu = \frac{1}{2}$), the rest of the analysis of COMP, DD, and SSS proceeds similarly to that done for the graphs in the typical set \eqref{eq:typical},\footnote{For SSS, we need to slightly strengthen the assumption $d = o(\sqrt{k})$ to $d = o\big(  \frac{\sqrt k}{\log n}\big)$; see the part of the proof following \eqref{eq:lim1}.} and yields analogous results with $O(k \log n)$ scaling, albeit slightly worse constant factors.  For GROTESQUE, the extension is even more immediate, since we did not use any characteristics of $\PP_G[Y=1]$ in its analysis.

\section{Challenges in the Analysis Compared to Standard Group Testing} \label{sec:differences}

Recall that the standard group testing problem concerns recovering a defective set $S \subseteq \{1,\dotsc,N\}$ of cardinality $K$ from a set of $N$ items, using a sequence of tests in groups of items \cite{Du93,Ald19}.  Each test returns one if there is at least one defective item in the test, and zero otherwise.  In this section, we highlight some of the main challenges arising in our analysis of COMP, DD, SSS, and GROTESQUE compared to their counterparts for standard group testing \cite{Cha11,Ald14a,Cai13}.

An immediate challenge is that in contrast with group testing, the analysis is not symmetric with respect to graphs having a given number of edges (e.g., the degree of each node also plays a major role).  Related to this issue is the fact that the events associated with including two different edges in a given test are {\em not} independent if those edges have a node in common.  As a a result, we frequently need to distinguish between error events for neighbors vs.~non-neighbors of a given node pair $(i,j)$.

For the COMP algorithm ({\em cf.}, Section \ref{sec:comp}), our analysis closely follows that of group testing after Lemma \ref{lem:typical} (establishing the asymptotic behavior of $\PP_G[Y=1]$) is established.  However, as seen in Appendix \ref{sec:pf_typical}, the proof of that lemma was in itself highly non-trivial.

For DD ({\em cf.}, Section \ref{sec:dd}), in addition to Lemma \ref{lem:typical}, additional effort is needed to handle the fact that different graphs lead to different probabilities of the key error events, which is not the case in group testing.  The first step (Section \ref{sec:dd_first}), distinguishes between the total number of non-edges in PE, and the number that actually share a node in common with a true edge, leading to several conditions on $t$ in \eqref{eq:dd_t1_1}--\eqref{eq:dd_t1_2} that luckily end up simplifying in Section \ref{sec:dd_combining}.  This also impacts the analysis of the second step (Section \ref{sec:dd_second}), where we require a careful analysis in \eqref{eq:qij_step1}--\eqref{eq:A2_step4} to characterize the key ``success event'' of including a given edge and no other pairs from the set $\PE$.

For SSS, ({\em cf.}, Section \ref{sec:sss}), we face similar difficulties in bounding the individual and pairwise events in \eqref{eq:deCaen2}, with a more delicate analysis leading to the remainder terms $\xi$ and $\xi'$ in \eqref{eq:xi} and \eqref{eq:pij_upper3}.  To complete the analysis after \eqref{eq:Perr_lb_2}, it is essential that these remainder terms not only behave as $o(1)$, but decay to zero {\em sufficiently fast}.

For GROTESQUE ({\em cf.}, Section \ref{sec:sublinear}), the multiplicity test ({\em cf.}, Section \ref{sec:mult_test}) is a fairly straightforward extension of that of standard group testing.  However, the location test ({\em cf.}, Section \ref{sec:loc_test}) is more novel.  In group testing, each item can be assigned a unique length-$L$ binary string, and the single defective item under consideration can trivially be identified using one test per bit. In our setting, we need to design tests that simultaneously identify two nodes, and when the two bits of the corresponding strings differ, we need to ensure that the bit assignments are done in a consistent manner across the $L$ indices ({\em cf.}, Step 3(b) in Section \ref{sec:loc_test}).  

\bibliographystyle{myIEEEtran}
\bibliography{JS_References}

\end{document}